\documentclass[%
 reprint,
 amsmath,amssymb,
 aps,
]{revtex4-2}
\usepackage[colorlinks=true,citecolor=blue,linkcolor=blue,urlcolor=blue]{hyperref}
\usepackage{graphicx}% Include figure files
\usepackage{dcolumn}% Align table columns on decimal point
\usepackage{bm}
\usepackage{amsthm}
\usepackage{mathtools}

\usepackage[colorlinks=true,linkcolor=blue,citecolor=blue,urlcolor=blue]{hyperref}

\newcommand{\id}{\mathbb{I}}

\newcommand{\TV}{\operatorname{TV}}
\newcommand{\dist}{\operatorname{dist}} 

\usepackage{algorithm}% http://ctan.org/pkg/algorithms
\usepackage{algpseudocode}% http://ctan.org/pkg/algorithmicx

\newcommand{\CZ}{\mathrm{CZ}}
\newcommand{\Rx}{R_x}
\newcommand{\Rz}{R_z}
\newcommand{\open}{\mathrm{open}}

\usepackage{tikz-cd} 
\usepackage{adjustbox}

\usepackage{tikz}						%geometric/algebraic description.
\usetikzlibrary{math, arrows,shapes.misc,snakes,
		       automata,backgrounds,
		       petri,topaths, decorations.pathmorphing, tikzmark}	

\usepackage{pgffor}

\usepackage{pgfplots}
\usetikzlibrary{pgfplots.groupplots}
\pgfplotsset{compat=1.11}
\usepgfplotslibrary{fillbetween}

\usepackage{tikz}
\usetikzlibrary{
  shapes,
  shapes.geometric,
	trees,
	matrix,
  positioning,
    pgfplots.groupplots,
  }

\PassOptionsToPackage{pdftex,hyperfootnotes=false,pdfpagelabels}{hyperref}

\newcommand{\ket}[1]{|#1\rangle}
\newcommand{\bra}[1]{\langle #1|}

\newcommand{\Cov}{\operatorname{Cov}}
\usepackage{tikz}
\usetikzlibrary{quantikz2}

\newtheorem{theorem}{Theorem}
\newtheorem{remark}{Remark}

\newtheorem{lemma}[theorem]{Lemma}

\newtheorem{definition}[theorem]{Definition}

\newtheorem{corollary}[theorem]{Corollary}

\newtheorem{observation}[theorem]{Observation}

\newcommand{\Tr}{\operatorname{Tr}}
\begin{document}

\preprint{APS/123-QED}

\title{Representational separation between unitary and channel quantum generative models via shared classical randomness at shallow depth}

\author{Arunava Majumder}
 \email{arunava.majumder@uibk.ac.at}
 \affiliation{%
University of Innsbruck, Department of
Theoretical Physics, Technikerstraße 21a, A-6020 Innsbruck, Austria.}
 \author{Marius Krumm}
\affiliation{%
University of Innsbruck, Department of
Theoretical Physics, Technikerstraße 21a, A-6020 Innsbruck, Austria.}
\author{Hendrik Poulsen Nautrup}
\affiliation{University of Innsbruck, Department of
Theoretical Physics, Technikerstraße 21a, A-6020 Innsbruck, Austria.}
\author{Hans J. Briegel}
\affiliation{University of Innsbruck, Department of
Theoretical Physics, Technikerstraße 21a, A-6020 Innsbruck, Austria.}

\begin{abstract}
Near-term quantum hardware limits circuit depth and often imposes geometrically local connectivity for quantum generative models, restricting the output distributions accessible to shallow unitary Born models. 
Introducing stochasticity into a unitary quantum Born model can improve the empirical generative performance of the resulting channel model and, for a restricted small-scale architecture, has been proven to represent a strictly larger family of distributions than its unitary counterpart. However, whether such randomness provides a provable separation at fixed shallow depth for arbitrarily large systems has remained open. Here, we show that shared classical randomness, a comparatively weak resource from entanglement theory, is sufficient to establish such a strict scalable representational separation over the corresponding shallow unitary Born model. More specifically, we augment bounded-connectivity shallow unitary circuits, followed by computational-basis measurements, with spatially separated local Pauli operations, whose joint application is controlled by a single classically sampled random bit. The resulting shallow-depth channel model generates long-range correlations in the classical output distribution that no purely unitary shallow-depth model with bounded connectivity can reproduce. For one-dimensional nearest-neighbour architectures, reproducing such distributions with a purely unitary model can require depth $\Omega(N)$ in the worst case. We further show that measurement-based quantum computation (MBQC) provides a natural implementation of the required shared classical randomness through suitable adaptation of the random measurement outcomes. Numerical experiments on MBQC-based generative models support the analytical results.
\end{abstract}

\maketitle

\section{Introduction}

Quantum computers are expected to provide advantages for certain computational tasks, with canonical examples including factorization~\cite{shor1999polynomial}, unstructured search~\cite{grover1996fast}, and sampling problems~\cite{arute2019quantum}. In the near-term setting, these prospects motivate variational quantum models~\cite{benedetti2019parameterized}, whose performance is constrained by realistic architectural resources such as connectivity, circuit depth, and measurement access~\cite{preskill2018quantum}. Among the variational quantum models, quantum generative models use a parameterized quantum circuit followed by a measurement to obtain a probability distribution over bit strings through Born's rule~\cite{liu2018differentiable,benedetti2019generative}. They are particularly promising because sampling from such distributions is known to be classically hard in general~\cite{RandomCircuitSampling, IQP}. However, the expressivity or representational power of these \textit{Born models}, in terms of the family of distributions they can generate, is tied not only to the Hilbert-space dimension, but also to the aforementioned architectural restrictions of the underlying circuit, including depth, connectivity, and available measurement operations.

Combining limited circuit depth with local connectivity restricts the propagation of information to a finite spatial range, preventing shallow circuits from establishing correlations between sufficiently distant qubits. \textit{Dynamic quantum circuits} can overcome this limitation by using mid-circuit measurements and classically controlled feedforward operations to establish long-range correlations~\cite{hwang2025distributed, baumer2024efficient, piroli2024approximating, EisertShallowMidCircuit}. Although powerful, these operations introduce additional experimental overhead: mid-circuit measurements are susceptible to measurement errors, while classical processing and communication can increase execution latency.

Several other approaches have explored stochastic or channel-based extensions of variational quantum models~\cite{majumder2024variational,majumder2026minimizing,heredge2025nonunitary,wen2026trainable}. In particular, Refs.~\cite{majumder2024variational,majumder2026minimizing} introduce classical randomness into the unitary model by classically processing the intrinsically random measurement outcomes of measurement-based quantum computation (MBQC). Separately, Ref.~\cite{majumder2024variational} proved that, for a restricted small-scale architecture, introducing this classical randomness yields a channel model that generates a strictly larger family of output distributions than the corresponding unitary model.
 However, whether such randomness establishes a strict scalable representational or expressivity advantage over a pure unitary Born model at shallow depth remains open. This motivates the central question of this work: 
 
 \textit{Can the family of output distributions generated by shallow local quantum circuits be strictly enlarged without introducing additional quantum resources in the form of circuit depth, long-range entangling gates, or adaptive mid-circuit measurements and feedforward?}

In this work, we show that shared classical randomness, a standard resource in entanglement theory~\cite{BuscemiLOSR, DSchmidLOSR}, is sufficient to establish a strict representational separation at shallow depth on any bounded-degree local architecture. We further allow the shared randomness to be tunable, making the resulting model suitable for generative learning tasks. As illustrated in
Fig.~\ref{fig:motivation}, a deep local circuit [Fig.~\ref{fig:motivation}(a)] can build correlations between distant qubits, whereas a shallow local circuit [Fig.~\ref{fig:motivation}(b)] cannot do so under the same connectivity constraints. Our construction
[Fig.~\ref{fig:motivation}(c)] keeps the circuit shallow, but augments it at an
intermediate layer with tunable shared randomness in the form of a stochastic
Pauli string whose support includes spatially separated qubits. The additional resource is therefore not a long-range quantum gate, but a shared classical random variable that determines in advance whether local Pauli gates will be applied simultaneously in distant regions of the circuit. Averaging over the random choices defines a stochastic channel model that can generate long-range correlations in its classical output distribution that the corresponding unitary model cannot reproduce at shallow depth, thereby establishing a strict separation between the two model families. In a one-dimensional nearest-neighbour architecture, reproducing these correlations with a purely unitary model requires a depth that grows linearly with the distance between the relevant regions. More generally, on any bounded degree
local architecture, the required unitary depth is
governed by the graph distance between these regions. In measurement-based quantum computation (MBQC), the measurement outcomes are inherently random, and suitable adaptation or retention of selected outcomes can naturally realize channel models with shared randomness without requiring extra resources.

Our results should not be interpreted as establishing a quantum computational advantage over classical methods. Rather, they establish a representational separation between two classes of quantum generative models subject to the same shallow, local quantum-resource constraints. One of these classes is given by standard shallow unitary Born machines with restricted connectivity.  The only additional resource of the other class is shared classical randomness,
used to coordinate Pauli operations on separated qubits. Thus, shared classical randomness alone can strictly enlarge the family of output distributions accessible to shallow, local quantum generative models without requiring additional quantum depth.

This manuscript is organized as follows. In Sec.~\ref{sec:gen_mod}, we introduce quantum generative modeling within the QCBM framework. In Secs.~\ref{sec:gen_bw_unit} and~\ref{sec:gen_bw_ch}, we define the local unitary model and its shared-randomness channel extension, respectively. Section~\ref{sec:mbqc_realizing_sh_ran} presents an MBQC realization of the shared-randomness channel, and Sec.~\ref{sec:mbqc_based_model} introduces the corresponding MBQC-based graph generative model. In Sec.~\ref{sec:analytical_res}, we present our analytical and numerical results. Section~\ref{sec:related_work} discusses related work. Finally, in the Supplementary Material, Appendices~\ref{app1:background}-\ref{app:num_details}, we provide detailed proofs and supporting calculations for the results established in Sec.~\ref{sec:analytical_res}.

\begin{figure*}[t]
\centering
\includegraphics[width=1.02\textwidth]{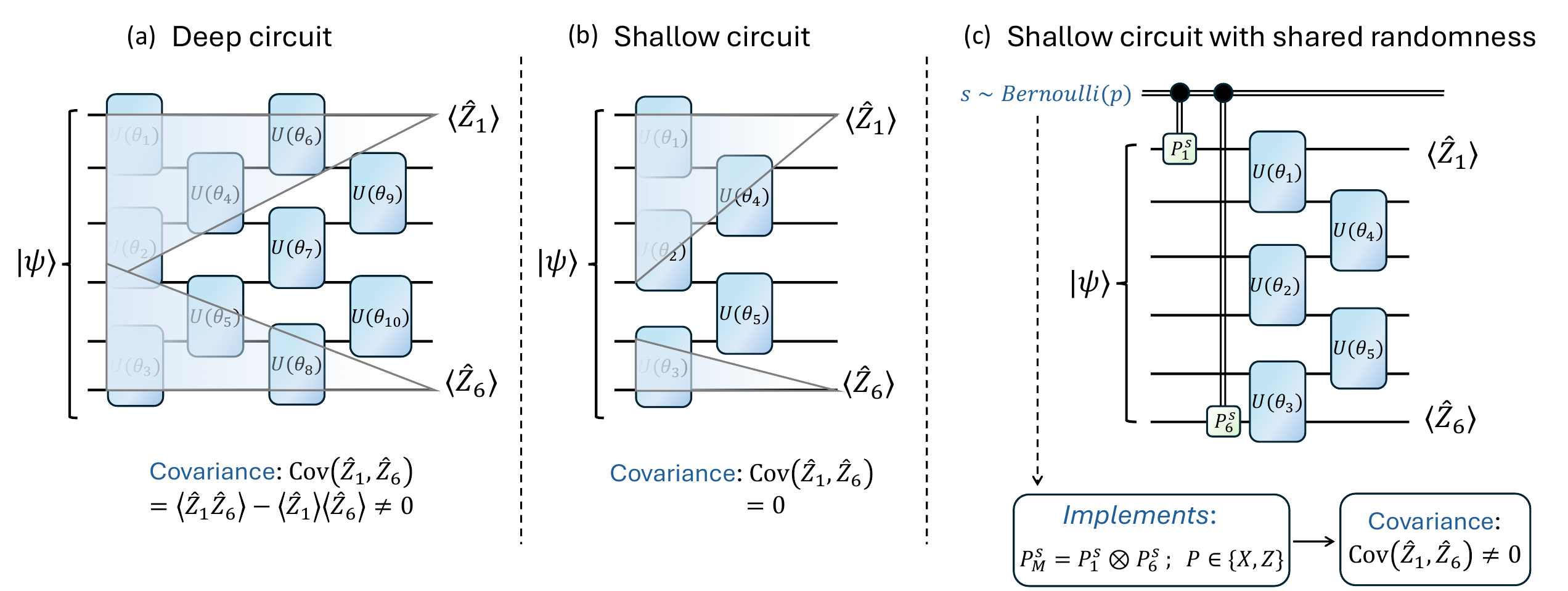}
\caption{\textbf{Long-range correlations from shared classical randomness:} For illustration, the figure shows a brickwall-type circuit on $N=6$ qubits acting on an input state $\ket{\psi}$, although the mechanism applies to any system size $N\geq 6$. The shaded regions denote the backward light cones of the measured observables $\hat Z_1$ and $\hat Z_6$. 
(a) In a sufficiently deep (linear depth) local circuit, the two light cones overlap, allowing the circuit to generate a nonzero correlation,
$\operatorname{Cov}(\hat Z_1,\hat Z_6)
=\langle \hat Z_1\hat Z_6\rangle
-\langle \hat Z_1\rangle\langle \hat Z_6\rangle\neq0$.
(b) For a shallow local circuit, the light cones remain disjoint. For a product input state, the corresponding expectation values factorize, and hence
$\operatorname{Cov}(\hat Z_1,\hat Z_6)=0$ for all model parameters.
(c) The same shallow circuit is augmented by one shared classical random bit $s\in\{0,1\}$ sampled from a Bernoulli distribution, which controls local Pauli operations on the two distant qubits and implements the stochastic string
$P_{\mathcal{M}}^s=P_1^s\otimes P_6^s$, with $P_q\in\{X,Z\}$, $P^0_q=\id$, and $q\in\{1,6\}$. The input state $\ket{\psi}$ is chosen such that $P_{\mathcal M}\ket{\psi}$ is not proportional to $\ket{\psi}$, ensuring that the Pauli string cannot be absorbed into the input state. Although no long-range quantum gate is introduced in (c), the shared random variable correlates the local responses and can produce
$\operatorname{Cov}(\hat Z_1,\hat Z_6)\neq0$ for suitable model parameters. Thus, shared classical randomness enables correlations that a purely unitary local circuit cannot generate at the same shallow depth.} 
\label{fig:motivation}
\end{figure*}

\section{Methods}

\subsection{Quantum generative modeling with shared randomness}
\subsubsection{Generative modelling}\label{sec:gen_mod}

\begin{figure}[t]
\centering
\includegraphics[width=0.51\textwidth]{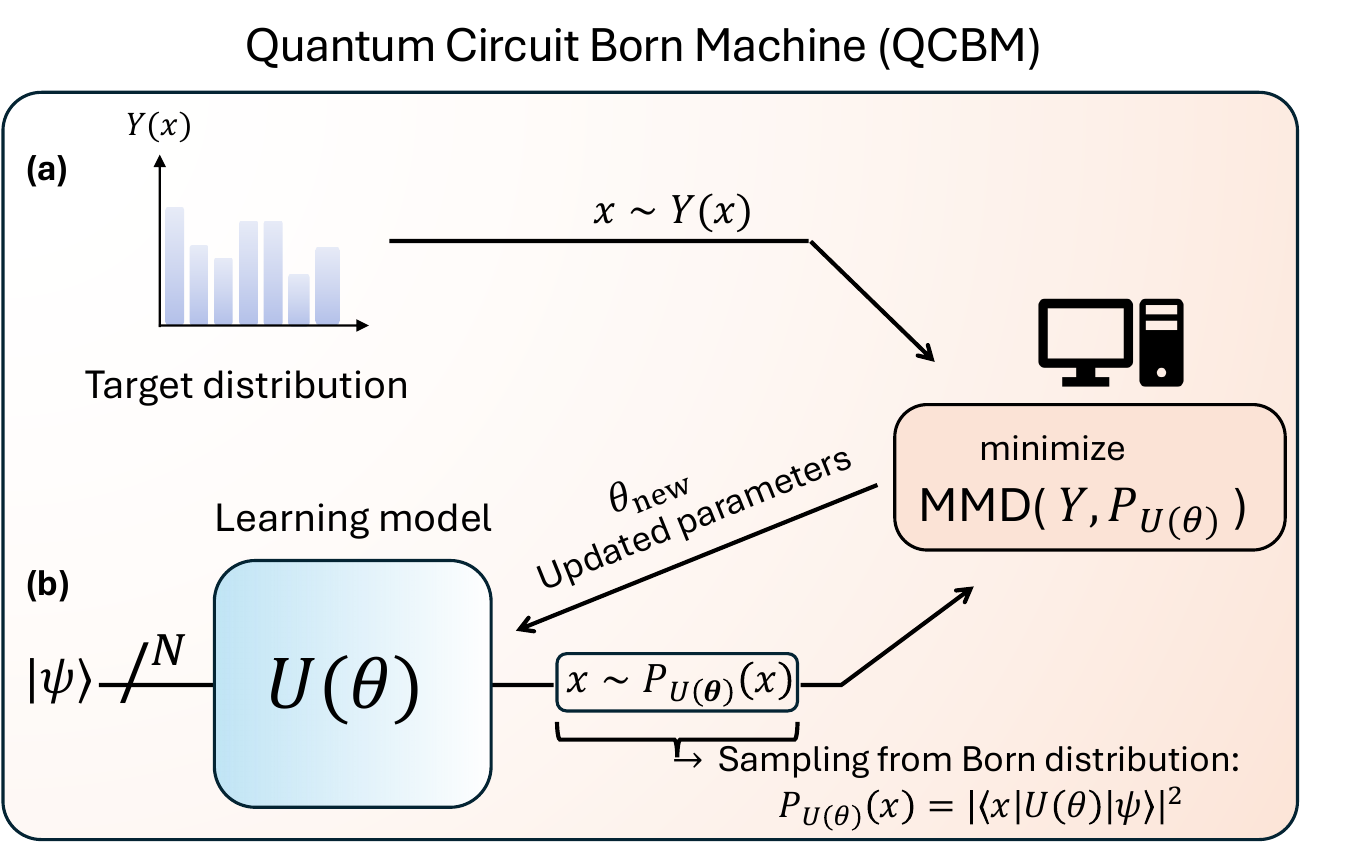}
\caption{\textbf{Quantum circuit Born machines for quantum generative modeling:} 
(a) Target distribution $Y(x)$ (usually unknown), from which samples are given, $x\sim Y$. (b) The learning model applies a parametrized quantum circuit $U(\bm\theta)$ to a fixed input state $\ket{\psi}$ and measures the resulting state in the computational basis. Each circuit execution produces a bit string sampled from the Born distribution $x\sim P_{U(\bm\theta)}$. Samples from the target and model distributions are used to evaluate the maximum mean discrepancy (MMD) loss between $Y$ and $P_{U(\bm\theta)}$. Gradients are estimated using the parameter-shift rule, after which a classical optimizer updates $\bm\theta\mapsto\bm\theta_{\rm new}$. The updated parameters are returned to the quantum circuit, and the procedure is repeated until the loss converges.} 
\label{fig:qgm}
\end{figure}

Generative learning is an unsupervised learning task that aims to model the probability distribution underlying a given data set. Generative models are commonly classified as \emph{explicit}, when they provide a tractable expression for the model distribution $P_{U(\bm\theta)}$, or \emph{implicit}, when they specify only a stochastic procedure for generating samples, $x\sim P_{U(\bm\theta)}$~\cite{mohamed2016learning}. We focus on implicit generative models. The training data consist of $M$ independent and identically distributed samples $\mathcal D=\{x_1,\ldots,x_{\mathcal{M}}\}$, where each $x_k\in\{0,1\}^N$ is a bit string drawn from an unknown target distribution $Y(x)$. The goal is to optimize the model parameters $\bm\theta$ such that the resulting distribution $P_{U(\bm\theta^\star)}(x)$ approximates $Y(x)$ to a desired accuracy, as illustrated in Fig.~\ref{fig:qgm}.

Classical generative models typically use neural networks with trainable weights to approximate a target distribution. In a quantum circuit Born machine (QCBM), a class of quantum generative models, the neural network is replaced by a parametrized quantum circuit composed of trainable single- and multi-qubit gates. Starting from a fixed input state $\ket{\psi}$, the circuit prepares the parametrized state $U(\bm\theta)\ket{\psi}$, which is measured in the computational basis to generate samples according to the Born rule, $P_{U(\bm\theta)}(x)=\left|\langle x|U(\bm\theta)|\psi\rangle\right|^2$. Each circuit execution produces one bit-string sample $x$, and repeated executions provide the samples required to estimate a suitable loss function and optimize the model parameters. A key feature of QCBMs is that they generate samples directly from the model distribution without requiring its explicit evaluation. The models studied in this paper follow the same sampling and optimization procedure shown in Fig.~\ref{fig:qgm}.

\subsubsection{Generic brickwall unitary circuit}\label{sec:gen_bw_unit}

We consider a shallow one-dimensional local quantum circuit on $N$ qubits, composed of $D$ layers
\begin{equation}\label{eq:bw_cir}
    U(\bm{\theta})
    =
    U_D(\bm{\theta}_D)\cdots U_1(\bm{\theta}_1),
\end{equation}
where
\begin{equation}
    \bm\theta=\{\bm\theta_k\}_{k=1}^D
\end{equation}
and $\bm\theta_k$ denotes the set of trainable parameters in layer $k$. For our purposes, we restrict to one-dimensional nearest-neighbor \textit{brickwall} circuits in which each gate acts on at most two adjacent qubits. Such a circuit consists of alternating sub-layers of parametrized two-qubit gates acting on disjoint even and odd nearest-neighbor bonds (see Fig.~\ref{fig:motivation}).
For example, one brickwall layer may be written as
\begin{equation}
    U_k(\bm \theta_k)
    =
    U^{(k)}_{\mathrm{odd}}(\bm{\theta}^{(k)}_{odd})
    U^{(k)}_{\mathrm{even}}(\bm{\theta}^{(k)}_{even}),
\end{equation}
where
\begin{equation}\label{eq:bricks}
    U^{(k)}_{\mathrm{even}}
    =
    \prod_{\substack{i=2\\ i\ \mathrm{even}}}^{N-1}
    G^{(k)}_{i,i+1}(\theta^{(k)}_{i,i+1}),
    ~~
    U^{(k)}_{\mathrm{odd}}
    =
    \prod_{\substack{i=1\\ i\ \mathrm{odd}}}^{N-1}
    G^{(k)}_{i,i+1}(\theta^{(k)}_{i,i+1}).
\end{equation}
and $G^{(k)}_{i,i+1}$ are parametrized two-qubit gates acting on qubits $(i,i+1)$ with $\theta^{(k)}_{i,i+1}\in\mathbb R$. Within each sub-layer, the gates act on disjoint pairs and can be applied in parallel. 
Between the even and odd sub-layers, the gates generally do not commute, so their ordering is part of the circuit definition. 
Single-qubit rotations can also be included in the generic parametrized two-qubit gates.

Throughout this paper, for the generic brickwall circuit we consider a fixed product input state
\begin{equation}\label{eq:input_st}
    \rho_{\rm in}
    =
    \bigotimes_{j=1}^N
    \rho_j,
\end{equation}
where each $\rho_j$ is an arbitrary single-qubit state. For a unitary model $U(\bm{\theta})$ as in Eq.~\eqref{eq:bw_cir}, measurement of the output state in the computational basis therefore produces the distribution
\begin{equation}\label{eq:unit_dist}
    P_{U(\bm{\theta})}(x)
    =
    \operatorname{Tr}
    \left[
        \ket{x}\!\bra{x}
        U(\bm{\theta})\rho_{\rm in}U(\bm{\theta})^\dagger
    \right],
    ~~ x\in\{0,1\}^N
\end{equation}

\subsubsection{Generic brickwall channel circuit}\label{sec:gen_bw_ch}
We now augment the unitary brickwall circuit with correlated stochastic Pauli operations. Let $P_{\mathcal{M}}$ be a Pauli string with
\begin{equation}
P_{\mathcal{M}}
=
\bigotimes_{j=1}^N Q_j,
\end{equation}
where ${\mathcal{M}}\subseteq[1,N]:=\{1,\ldots,N\}$ denotes the support of the Pauli string, $Q_j\in\{X,Z\}$ for $j\in {\mathcal{M}}$, and $Q_j=\id$ for $j\notin {\mathcal{M}}$. A single classical Bernoulli variable $s\in\{0,1\}$ determines whether the entire Pauli string is applied, such that
\begin{equation}
P_{\mathcal{M}}^0=\id,
\qquad
P_{\mathcal{M}}^1=P_{\mathcal{M}}.
\end{equation}
Thus, all local Pauli operators supported on ${\mathcal{M}}$ are controlled by the same random variable $s$, producing a spatially correlated stochastic Pauli operation. 

Consider an insertion slot $l\in\{1,\ldots,D-1\}$ between the circuit layers $U_l(\bm{\theta}_l)$ and $U_{l+1}(\bm{\theta}_{l+1})$. The unitary associated with branch $s$ is defined as
\begin{equation}\label{eq:min_channel_def}
U^{(s)}(\bm{\theta})
:=
U_D(\bm{\theta}_D)
\cdots
U_{l+1}(\bm{\theta}_{l+1})
P_{\mathcal{M}}^s
U_l(\bm{\theta}_l)
\cdots
U_1(\bm{\theta}_1).
\end{equation}
The Pauli string is applied with probability $\Pr(s=1)=p$, where $p\in[0,1]$, and is omitted with probability $\Pr(s=0)=1-p$. Equivalently,
\begin{equation}
\Pr(s)
=
p^s(1-p)^{1-s}.
\end{equation}
We refer to $p$ as the \textit{application probability}. The values $p=0$ and $p=1$ correspond to deterministic unitaries which we refer to as branches depending on $s\in\{0,1\}$, whereas $p\in(0,1)$ defines a nontrivial stochastic channel. 

The resulting channel model $\mathcal E_{\bm \theta,p}(\rho_{\rm in})$,  acting on the input state $\rho_{\rm in}$ becomes 
\begin{equation}\label{eq:channel_model}
    \mathcal E_{\bm \theta, p}(\rho_{\rm in})
    =
    \sum_{s\in\{0,1\}}
    \Pr(s)\,
    U^{(s)}(\bm{\theta})\rho_{\rm in} U^{(s)}(\bm{\theta})^\dagger,
\end{equation}
where $U^{(0)}(\bm \theta)$ denotes the branch without any correlated Pauli gates. Measuring the channel output in the computational basis gives the distribution
\begin{equation}\label{eq:channel_dist}
    P_{\mathcal E_{\bm{\theta},p}}(x)
    = \sum_{s\in\{0,1\}}
    \Pr(s)
    \operatorname{Tr}
    \left[
        \ket{x}\!\bra{x}
        U^{(s)}(\bm{\theta})\rho_{\rm in}U^{(s)}(\bm{\theta})^{\dagger}]
    \right]
\end{equation}
The construction extends directly to multiple stochastic Pauli strings inserted at different positions within the circuit. Let $\{P_{{\mathcal{M}}_i}\}_{i=1}^L$ be a set of $L$ Pauli strings, where ${\mathcal{M}}_i\subseteq[1,N]$. Each Pauli string $P_{{\mathcal{M}}_i}$ is controlled by a Bernoulli variable $s^i\in\{0,1\}$ with
\begin{equation}
\Pr(s^i=1)=p^i,
\qquad
\Pr(s^i=0)=1-p^i,
\end{equation}
where $p^i\in[0,1]$. We assume that the variables $\{s^i\}_{i=1}^L$ are mutually independent and denote the corresponding application probabilities by the set $\bm p=\{p^i\}_{i=1}^L$. Although the variables controlling different strings are independent, each individual string $P_{{\mathcal{M}}_i}^{s^i}$ represents a correlated operation because the same variable $s^i$ controls all Pauli operators supported on ${\mathcal{M}}_i$.

Each Pauli string $P_{{\mathcal{M}}_i}$, with $i\in\{1,\ldots,L\}$, is assigned to an intermediate layer denoted by an insertion slot $l_i\in\{1,\ldots,D-1\}$. The assignment $\{l_i\}_{i=1}^L$ is fixed as part of the circuit architecture. Different Pauli strings may be assigned to different slots, multiple Pauli strings may share the same slot, and some slots may contain no Pauli strings. Having several Pauli strings for the same slot allows to more directly accommodate correlations in which non-overlapping subsets of qubits are correlated within the same subset, but not across the subsets.

For a branch $\bm s=(s^1,\ldots,s^L)\in\{0,1\}^L$ and fixed assignments $\{l_i\}_{i=1}^L$, define
\begin{equation}\label{eq:multi_pauli_corr}
B_l(\bm s)
:=
\prod_{\substack{1\leq i\leq L,\\ l_i=l}}
P_{{\mathcal{M}}_i}^{s^i},
\qquad
l\in\{1,\ldots,D-1\}.
\end{equation}
When iterating over $1\leq i\leq L$, the condition $l_i=l$ selects precisely those Pauli strings whose assigned insertion slot is $l$ in the circuit definition, that is, those inserted between $U_l$ and $U_{l+1}$. Thus, $B_l(\bm s)$ is the product of all Pauli strings assigned to the insertion slot between $U_l$ and $U_{l+1}$. The product is taken in a fixed order, and $B_l(\bm s)=\id$ if no Pauli string is assigned to slot $l$. The branch $\bm s$ dependent unitary is

\begin{equation}\label{eq:unit_with_byp}
    U^{(\bm s)}(\bm\theta)
=
U_D(\bm\theta_D)
B_{D-1}(\bm s)
U_{D-1}(\bm\theta_{D-1})
\cdots
B_1(\bm s)
U_1(\bm\theta_1).
\end{equation}
and the probability of branch $\bm s$ being realized is
\begin{equation}\label{eq:branch_prob_decomp}
\pi_{\bm p}(\bm s)
=
\prod_{i=1}^L
(p^i)^{s^i}(1-p^i)^{1-s^i}=\prod_{i=1}^L\pi_{p^i}(s^i).
\end{equation}
The induced channel model can be written as
\begin{equation}\label{eq:channel_model_set}
    \mathcal E_{\bm\theta,\bm p}(\rho_{\rm in})
    =
    \sum_{\bm s\in\{0,1\}^L}
    \pi_{\bm p}(\bm s)\,
    U^{(\bm s)}(\bm\theta)
    \rho_{\rm in}
    U^{(\bm s)}(\bm\theta)^\dagger .
\end{equation}
and the corresponding Born distribution is
\begin{equation}\label{eq:channel_dist_all}
    P_{\mathcal E_{\bm\theta,\bm p}}(x)
    =
    \sum_{\bm s\in\{0,1\}^L}
    \pi_{\bm p}(\bm s)\,
    \Tr\!\left[
        \ket{x}\!\bra{x}\,
        U^{(\bm s)}(\bm\theta)
        \rho_{\rm in}
        U^{(\bm s)}(\bm\theta)^\dagger
    \right].
\end{equation}
where $\bm \theta=\{\bm\theta_l\}_{l=1}^D$ with $\bm \theta_l$ containing the set of parameters for layer $l$, and $\bm p=\{p^i\}_{i=1}^L$ contains the probabilities with which the correlated Paulis are applied.

\begin{definition}[Minimal correlated channel model] \label{def:min_channel} The minimal correlated channel model is obtained by inserting a single stochastic Pauli string $P_{\mathcal{M}}^s$ at an intermediate layer $l$ of the local quantum circuit defined in Eq.~\eqref{eq:min_channel_def}, where the support ${\mathcal{M}}$ contains two distinct sites $a,b\in {\mathcal{M}}$, i.e., $P^s_{\mathcal{M}}=P^s_a\otimes P^s_b$ with $P_q\in\{X,Z\}$, $P_q^0=\id$, and $q\in \{a,b\}$. The resulting channel, for an input state $\rho_{\rm in}$ is
\begin{align}\label{eq:min_channel_def_split}
\mathcal E_{\bm{\theta},p}(\rho_{\rm in})
=&
(1-p)
U^{(0)}(\bm{\theta})\rho_{\rm in}U^{(0)}(\bm{\theta})^\dagger\nonumber
+\\
&p
U^{(1)}(\bm{\theta})
\rho_{\rm in}
U^{(1)}(\bm{\theta})^\dagger
\end{align}
where $U^{(s)}(\bm \theta)$ with $s\in\{0,1\}$ follows from Eq.~\eqref{eq:min_channel_def}.
\end{definition}

\subsection{Learning model}\label{sec:mbqc_based_model}

\begin{figure}[t]
\centering
\includegraphics[width=0.48\textwidth]{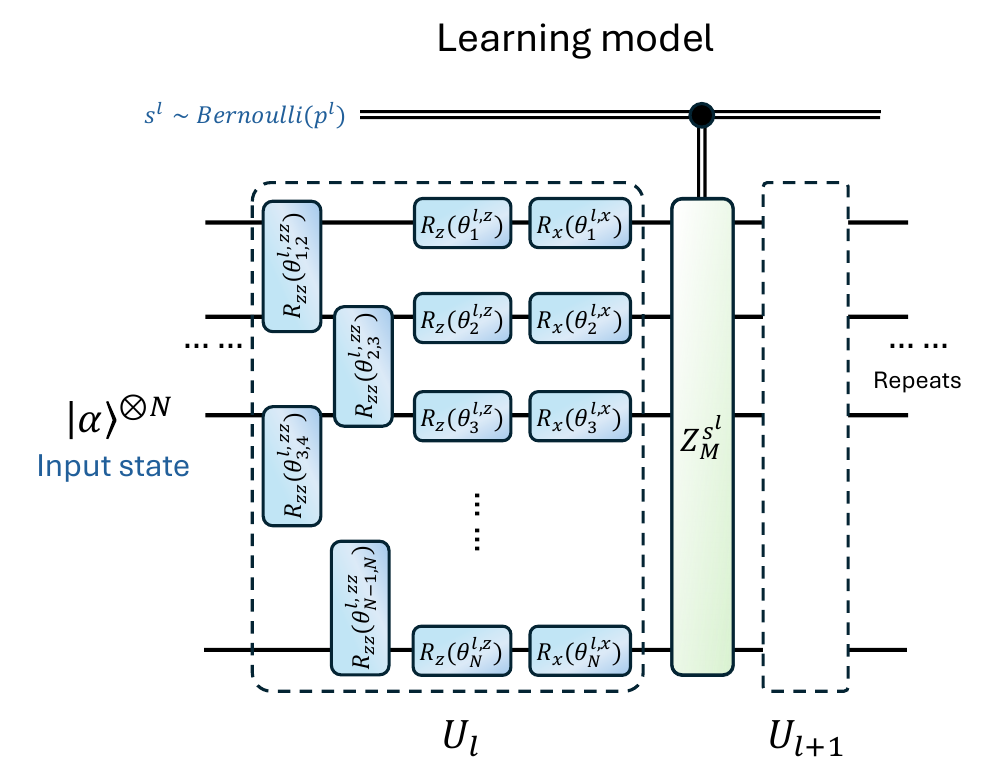}
\caption{\textbf{Learning model with shared randomness:} A depth-$D$ one-dimensional nearest-neighbor circuit acts on the product input state $\ket{\alpha}^{\otimes N}$. Each parametrized layer $U_l(\bm{\theta}_l)$ consists of nearest-neighbor $R_{zz}$ rotations followed by local $R_z$ and $R_x$ rotations. Between layers $U_l$ and $U_{l+1}$, a shared classical bit $s^l\sim\operatorname{Bernoulli}(p^l)$ allows the stochastic Pauli string $Z_{\mathcal M}^{s^l}$ to act on a subset $\mathcal M\subseteq[1,N]$ of qubits, such that the local $Z$ operations are either applied jointly or all omitted. Repeating the layer structure generates the full depth-$D$ learning model while the Pauli
string shown here is inserted only at the designated slot $l$.} 
\label{fig:learning_model_main}
\end{figure}

As an explicit learning model, we consider a depth-$D$ one-dimensional nearest-neighbor circuit on $N$ qubits with open boundary conditions as in Eq.~\eqref{eq:bw_cir},
where each layer $k\in[1,D]$ in the circuit can be written as 
\begin{equation}\label{eq:model_num_main}
U_k(\bm{\theta}_k)=\bigotimes_{i=1}^N R_x(\theta^{k,x}_{i}) \bigotimes_{i=1}^N R_z(\theta^{k,z}_{i})\prod_{i=1}^{N-1}R_{zz}(\theta^{k,zz}_{i,i+1}) 
\end{equation}
as shown in Fig.~\ref{fig:learning_model_main} with
\begin{equation}
R_x(\theta)=e^{-i\theta X/2},
~
R_z(\theta)=e^{-i\theta Z/2},
~
R_{zz}(\theta)=e^{-i\theta Z\otimes Z/2}.
\end{equation}
Here, $R_x$ and $R_z$ are single-qubit Pauli-rotations, while $R_{zz}$ is a two-qubit entangling Pauli-rotation acting on nearest-neighbor qubits, and the $k^{\rm th}$ layer parameters are
\begin{equation}\label{eq:layer_params_main}
\bm{\theta}_k
=
\left(
\{\theta^{k,x}_i\}_{i=1}^{N},
\{\theta^{k,z}_i\}_{i=1}^{N},
\{\theta^{k,zz}_{i,i+1}\}_{i=1}^{N-1}
\right).
\end{equation}
For an insertion slot after layer $l\in[1,D-1]$, write
\begin{equation}
U(\bm{\theta})
=
U_{>l}(\bm{\theta})U_{\leq l}(\bm{\theta}).
\end{equation}
where $U_{\leq l}(\bm{\theta})$ contains all circuit layers up to layer $l$, while $U_{>l}(\bm{\theta})$ contains the subsequent layers.
A shared bit $s^l\in\{0,1\}$ is sampled according to
$\Pr(s^l=1)=p^l$, and the stochastic Pauli string
\begin{equation}
Z_{\mathcal M}^{s^l}
:=
\bigotimes_{j\in\mathcal M} Z_j^{s^l},
\qquad
\mathcal M\subseteq[1,N],
\end{equation}
is inserted at this slot, between layer $l$ and $l+1$, as shown in Fig.~\ref{fig:learning_model_main}. The branch unitary becomes
\begin{equation}
U^{(s^l)}(\bm{\theta})
=
U_{>l}(\bm{\theta})
Z_{\mathcal M}^{s^l}
U_{\leq l}(\bm{\theta}).
\end{equation}

The resulting channel is
\begin{align}
\mathcal E_{\bm{\theta},p^l}(\rho_{\rm in})
={}&
(1-p^l)\,
U^{(0)}(\bm{\theta})\rho_{\rm in}U^{(0)}(\bm{\theta})^\dagger
\nonumber\\
&+
p^l\,
U^{(1)}(\bm{\theta})\rho_{\rm in}U^{(1)}(\bm{\theta})^\dagger,
\label{eq:num_channel_main}
\end{align}
and computational-basis measurements give, as output probability,
\begin{equation}
P_{\mathcal E_{\bm{\theta},p^l}}(\bm{x})
=
\operatorname{Tr}\!\left[
\ket{\bm{x}}\!\bra{\bm{x}}\,
\mathcal E_{\bm{\theta},p^l}(\rho_{\rm in})
\right].
\end{equation}
We also use a fixed product input state
\begin{equation}\label{eq:input_st_r}
    \rho_{\rm in}
    =
    \bigotimes_{j=1}^N
    R_x(\alpha)\ket{0}\!\bra{0}R_x(\alpha)^\dagger,
\end{equation}
with a fixed angle $\alpha\notin\frac{\pi}{2}\mathbb{Z}$. For our numerical simulations, we explicitly set $\alpha=\frac{\pi}{4}$.

Propagating the inserted Pauli string through $U_{>l}(\bm{\theta})$ leaves commuting Pauli rotations unchanged while flipping the signs of the angles of anti-commuting rotations, as illustrated in Sec.~\ref{def:input_byp_model} (also see Fig.~\ref{fig:two_learning_models}) in the appendix. The choice of a correlated Pauli-$Z$ string ensures that the shared randomness is encoded entirely in correlated sign flips of the non-Clifford angles of $R_x$ gates. These branch-dependent angle flips modify the circuit dynamics and hence the output distribution non-trivially, while the remaining propagated $Z$-string is diagonal in the computational basis and does not alter the final measured bit string. More generally, correlated Pauli-$X$ or mixed Pauli strings may also be considered, although their direct action on computational-basis outcomes must then be taken into account. Later, we show in Sec.~\ref{sec:end_pauli_corr} that any resulting non-$Z$ Pauli string at the circuit output can be corrected, leaving the branch-dependent intermediate angle flips as the effective action of the stochastic Pauli string.

The learning model defined in Eq.~\eqref{eq:num_channel_main} is used in the numerical analysis in Sec.~\ref{sec:num_results}. We note, however, that the analytical results presented in Sec.~\ref{sec:analytical_res} do not depend on this specific architecture and can be extended to more general finite-range local architectures.

\subsection{Native correlated channel models in variational measurement-based quantum computation }\label{sec:mbqc_realizing_sh_ran}

\subsubsection{Introduction to variational measurement-based quantum computation}\label{sec:intro_to_vmbqc}
Here we describe how shared classical randomness, represented by a \textit{stochastic Pauli string}, or \textit{correlated stochastic Pauli operations}, can be realized natively within the MBQC framework. In MBQC, universal quantum computation is performed through adaptive single-qubit measurements on an entangled resource state~\cite{raussendorf2001one,briegel2001persistent}. To illustrate this mechanism, we use the standard cluster state, prepared by initializing qubits in the state $\ket{+}$ on a rectangular lattice and applying nearest-neighbor controlled-$Z$ gates. The construction is not restricted to cluster states: other graph states supporting MBQC can likewise realize unitary circuits augmented with shared randomness. In the left panels of Fig.~\ref{fig:mbqc_explain}~(a,b), we consider a cluster state of size $N\times(D+1)=4\times 3$ with open boundary conditions, so that boundary qubits are not connected along the spatial direction. The first $D$ columns are measured in the $XY$ plane, while the final column forms the output layer and is measured in the computational $Z$ basis.

We label the measured qubits by $(l,j)$, where $l\in[1,D]$ is the layer index
and $j\in[1,N]$ is the spatial index. The measurement on qubit $j$ in layer
$l$ is performed in the basis
\begin{equation}
\ket{\pm_{\theta_j^l}}
=
\frac{1}{\sqrt{2}}
\left(
\ket{0}\pm e^{-i\theta_j^l}\ket{1}
\right),
\end{equation}
where $\theta_j^l$ determines the corresponding gate in the induced circuit
model~\cite{browne2006one}.

An individual measurement in this basis gives either outcome $+1$ or $-1$, each with
probability $1/2$. The $+1$ outcome implements the intended operation, while
the $-1$ outcome introduces an additional Pauli-$Z$ byproduct, since $
    \ket{-_{\theta}}\bra{-_{\theta}}
    =
    Z\ket{+_{\theta}}\bra{+_{\theta}}Z .
$
In standard MBQC, these measurement-induced byproducts are compensated by adapting subsequent measurement angles according to earlier measurement outcomes. A classical controller records the outcomes, determines the resulting byproduct operators, and updates later measurement bases through feedforward. This procedure is applied throughout the resource state, and the final readout is interpreted accordingly, so that the implemented logical transformation is independent of the stochastic measurement outcomes~\cite{majumder2024variational}.

For the cluster state considered here, MBQC with $XY$-plane measurements admits an equivalent circuit representation (see Fig.~\ref{fig:mbqc_explain}) consisting of trainable $Z$-axis
rotations, stochastic Pauli-$Z$ byproducts, and a fixed Clifford layer
$T_c$ defining a Clifford quantum cellular automaton (CQCA)~\cite{poulsen2024measurement,majumder2024variational,
majumder2026minimizing}. With open boundary conditions, we take
\begin{equation}
    T_c
    =
    H^{\otimes N}
    \prod_{j=1}^{N-1}CZ_{j,j+1}.
\end{equation}
and using the convention
\begin{equation}
    R_z(\theta_j^l)
    :=
    \exp\!\left(-i\theta_j^l Z_j/2\right),
\end{equation}
the circuit branch corresponding to measurement outcomes
$\bm s=\{s_j^l\}_{j,l}$ can be written as
\begin{equation}\label{eq:pure_unit}
    U^{(\bm s)}(\bm\theta)
    =
    \prod_{l=D,\ldots,1}
    \left[
        T_c
        \prod_{j=1}^{N}
        Z_{j}^{s_j^l}
        \exp\!\left(-i\theta_j^l Z_{j}/2\right)
    \right],
\end{equation}
where $s_j^l=0$ corresponds to the $+1$ measurement outcome and $s_j^l=1$ to the $-1$
outcome. This binary variable $s_j^l\in\{0,1\}$ decides whether a byproduct $Z^{s_j^l}_j$ should appear at position $(l,j)$ in the circuit picture (Fig.~\ref{fig:mbqc_explain}). The Eq.~\eqref{eq:pure_unit} gives rise to a variational model of MBQC, called variational MBQC (VMBQC)~\cite{majumder2024variational}. 

% As CQCAs map Pauli operators to Pauli operators under conjugation, for the open chain, its action on the single-qubit Pauli generators is
% %
% \begin{equation}
% \label{eq:T-func_3}
%     \begin{split}
%         T_c(X_1) &:= T_cX_1T^\dagger_c = Z_1 X_{2}\\
%         T_c(X_i) &:= T_cX_iT^\dagger_c = X_{i-1} Z_i X_{i+1},~~\mathrm{if}~i\in[2,N-1] \\
%         T_c(X_N) &:= T_cX_NT^\dagger_c = Z_N X_{N-1}\\
%         T_c(Z_i) &:= T_cZ_iT^\dagger_c = X_i 
%     \end{split}
% \end{equation} 

While standard MBQC always corrects byproducts, Refs.~\cite{majumder2024variational,majumder2026minimizing} allow for probabilistic correction of these byproducts to elevate the unitary VMBQC to a channel model. Let
$c_j^l\in\{0,1\}$ denote the correction decision, with
$c_j^l=1$ meaning that the byproduct is corrected and $c_j^l=0$ meaning that it
is retained or not corrected. Assuming that $c_j^l$ is sampled independently of the measurement outcome $s_j^l$, with
\begin{equation}
    \Pr(c_j^l=1)=p_j^l,
\end{equation}
then the retained byproduct bit is
\begin{equation}
    \widetilde{s}_j^l
    =
    (1-c_j^l)s_j^l .
\end{equation}
Since measurement outcomes are random and satisfy $\Pr(s_j^l=1)=1/2$, the retained byproduct probability is
\begin{equation}
\Pr(\widetilde{s}_j^l=1)
=
\frac{1}{2}(1-p_j^l)
\leq
\frac{1}{2}.
\end{equation}
This gives rise to a family of quantum channels
\begin{align}\label{eq:mbqc_ch}
\mathcal{E}_{(\bm{\theta}, \bm{p})}[\rho]
=
\sum_{\widetilde{\bm{s}}\in\{0,1\}^{N\times D}}
p(\widetilde{\bm{s}})
~U^{(\widetilde{\bm{s}})}(\bm\theta)
~\rho~
U^{(\widetilde{\bm{s}})}(\bm\theta)^{\dagger},
\end{align}
where $\rho=(\ket{+}\bra{+})^{\otimes N}$ is the input state. As the byproducts are independent of each other, the branch probability factorizes as $p(\bm{\tilde{s}})=\prod_{l,j}p(\tilde{s}_j^l)=\prod_{l,j}\left(\frac{1+(-1)^{\tilde{s}^l_j}p^l_j}{2}\right)$, with $\bm{\tilde{s}}$ being an $N\times D$ binary matrix with entries $\tilde{s}^j_i\in \{0,1\}$, and $p^j_i$ is the trainable correction probability for the qubit at $(i,j)$-location. In VMBQC, this channel character arises natively from random measurement-induced byproducts, whose retention or correction is governed entirely by classical processing and feedforward of the measurement outcomes~\cite{majumder2024variational}.

While it was shown in Refs.~\cite{majumder2024variational} that the trainable correction probability can give rise to models that are more expressive than the unitary model (even in a very restricted setting~\cite{majumder2026minimizing}), it has a rather strong shortcoming: this correction-only mechanism can suppress naturally occurring byproducts but cannot make any retained byproduct occur with probability greater than $1/2$.

In the present work, we allow more general classical control. In addition to deciding whether to retain a measurement-induced byproduct when $s_j^l=1$, the controller may deliberately introduce the same effective byproduct when $s_j^l=0$ through an appropriate modification of the subsequent feedforward rules. This allows the effective byproduct probability $\Pr(s_j^l=1)$ to vary over the full interval $[0,1]$. Moreover, a single classical controller can use the same sampled variable to apply a correlated Pauli byproduct string simultaneously to several spatially separated qubits. We illustrate these constructions using Fig.~\ref{fig:mbqc_explain}. The next Sec.~\ref{sec:ind_ran} introduces how one can deliberately add byproducts independently to individual qubits. Then, in Sec.~\ref{sec:mbqc_realization_of_co_ran}, we show the construction of shared randomness in the form of correlated byproducts on a subset of qubits similar to Eq.~\eqref{eq:channel_model_set} but implemented in VMBQC natively.

\begin{figure}[t]
\centering
\includegraphics[width=0.48\textwidth]{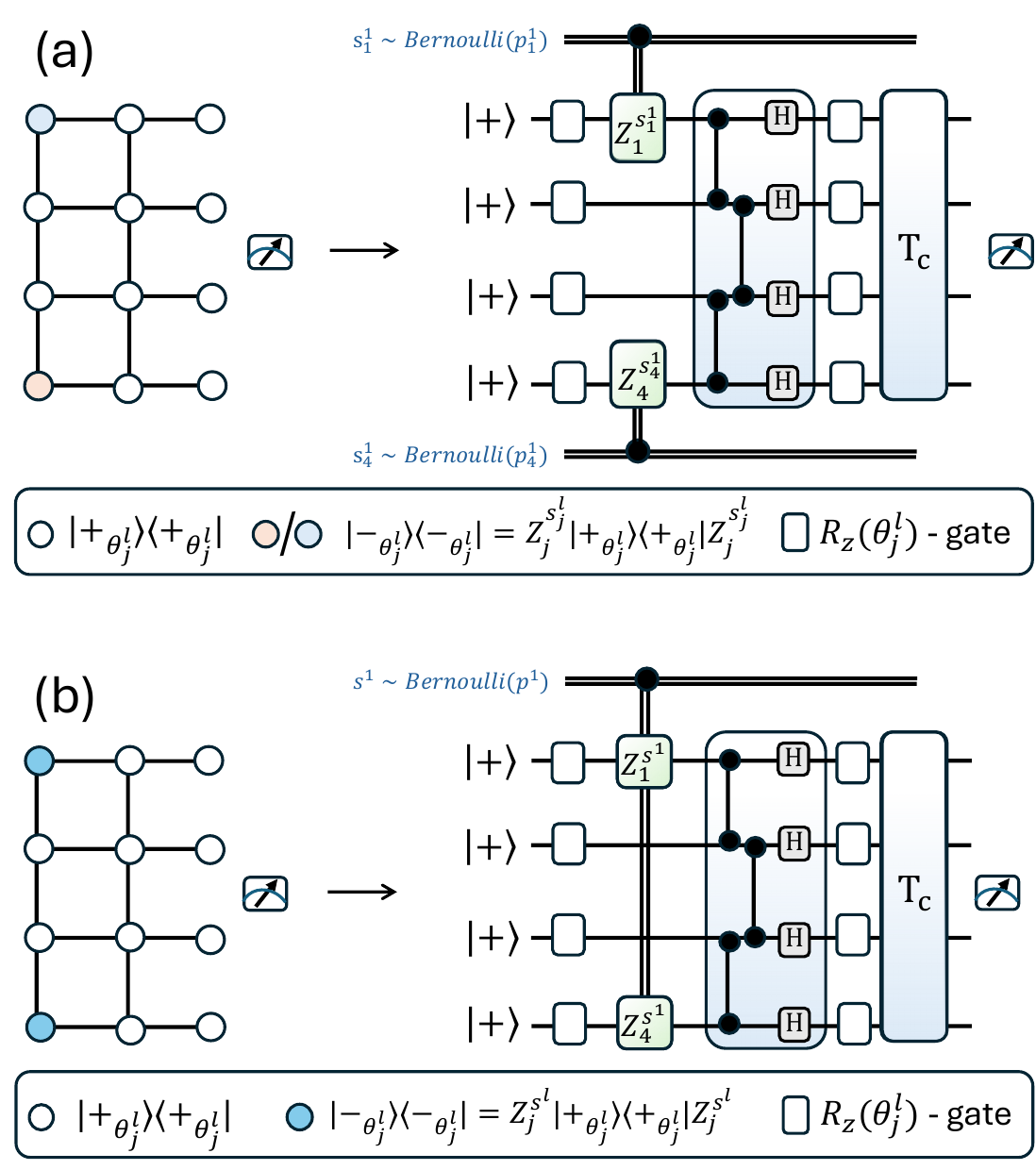}
\caption{\textbf{Independent and shared randomness within MBQC:} For illustration, we consider a $4\times 3$ cluster state, where the first two columns are measured in the $XY$ plane, and the final column forms the readout layer. The corresponding circuit representation contains parametrized $R_z(\theta_j^l)$ rotations, nearest-neighbor $CZ$ gates, and Hadamard gates. The colored boundary qubits indicate measurement-induced $Z$ byproducts that are either corrected by feedforward or deliberately added through anti-correction, all other qubits are corrected with unit probability. 
(a) The two boundary qubits are corrected or anti-corrected independently. The corresponding circuit therefore contains independent byproducts $Z_1^{s_1^1}$ and $Z_4^{s_4^1}$, controlled by independent Bernoulli variables $s_1^1\sim\operatorname{Bernoulli}(p_1^1)$ and $s_4^1\sim\operatorname{Bernoulli}(p_4^1)$, respectively, where $s_j^l=0$ denotes correction, $s_j^l=1$ denotes anti-correction of measurement-induced byproduct on qubit $j$ at layer $l$, and $\Pr(s_j^l=1)=p_j^l\in[0,1]$ as in Eq.~\eqref{eq:fig_3_a}. (b) The feedforward rule is chosen in a way so that the two boundary qubits (blue) are corrected or anti-corrected jointly. Their retained byproducts are therefore controlled by a single shared binary variable $s^1\sim\operatorname{Bernoulli}(p^1)$ (Eq.~\eqref{eq:fig_3_b}), yielding the stochastic Pauli string $Z_1^{s^1}\otimes Z_4^{s^1}$ in the equivalent circuit on the right.} 
\label{fig:mbqc_explain}
\end{figure}

\subsubsection{Independent randomness from effective classical control}
\label{sec:ind_ran}

We first describe the independent-randomness version of the construction within the MBQC framework.
For each measured qubit $(l,j)$, introduce an effective classical control variable
$c_j^l\in\{-1,+1\}$. Note that the range of this variable here is different from the variable $c_j^l\in\{0,1\}$ introduced in Sec.~\ref{sec:intro_to_vmbqc}. After obtaining the raw measurement outcome on qubit $(l,j)$, the choice $c_j^l=+1$ corrects the effective measurement-induced byproduct, whereas $c_j^l=-1$ produces the corresponding Pauli $Z$-byproduct. The latter may equivalently be interpreted as an \textit{anti-correction}, implemented by modifying the subsequent measurement angles if there is no measurement-induced byproduct. Thus, anti-correction can be viewed as the same operation as correction that treats $s=0$ as the ``incorrect" measurement outcome. Now, we define the
effective byproduct bit
\begin{equation}
    s_j^l
    :=
    \frac{1-c_j^l}{2}.
\end{equation}
Thus $s_j^l=0$ for $c_j^l=+1$ and $s_j^l=1$ for $c_j^l=-1$. Choosing
\begin{equation}
    \Pr(c_j^l=-1)=p_j^l,
    \qquad
    \Pr(c_j^l=+1)=1-p_j^l,
\end{equation}
is equivalent, in the induced circuit model of Fig.~\ref{fig:mbqc_explain} (a) (right),
to treating $s_j^l$ as a Bernoulli variable with tunable probability
\begin{equation}\label{eq:fig_3_a}
    s_j^l\sim \mathrm{Bernoulli}(p_j^l),
    ~~
    \Pr(s_j^l=1)=p_j^l,
    ~
    \Pr(s_j^l=0)=1-p_j^l .
\end{equation}
Operationally, the induced circuit applies the Pauli byproduct
$Z_j^{s_j^l}$ with probability $p_j^l$. Unlike the correction-only mechanism
described above in Eq.~\eqref{eq:mbqc_ch}, this effective control scheme allows
$\Pr(s_j^l=1)$ to range over the full interval $[0,1]$.

For the example in Fig.~\ref{fig:mbqc_explain}~(a), we use two independent
classical variables controlling the two byproducts $Z_1^{s_1^1}$ and
$Z_4^{s_4^1}$ in the circuit, with
\begin{equation}
    s_1^1\sim \mathrm{Bernoulli}(p_1^1),
    \qquad
    s_4^1\sim \mathrm{Bernoulli}(p_4^1).
\end{equation}
These two effective byproduct bits are sampled independently. The
corresponding branch circuit is
\begin{align}
    U^{(\bm s)}(\bm\theta)
    &=
    \left(
        T_c
        \prod_{j=1}^{4}
        \exp\!\left(-i\theta_j^2 Z_j/2\right)
    \right).\\
    &
    \left(
        T_c
        \left(Z_{1}^{s_1^1}\otimes Z_{4}^{s_4^1}\right)
        \prod_{j=1}^{4}
        \exp\!\left(-i\theta_j^1 Z_j/2\right)
    \right),
\end{align}
where $\bm s=(s_1^1,s_4^1)$. Making the probabilities $p_1^1$ and $p_4^1$
trainable gives the channel family
\begin{align}
\label{eq:independent_mbqc_channel_example}
    \mathcal{E}_{(\bm{\theta}, \bm{p})}[\rho]
    =
    \sum_{s_1^1,s_4^1\in\{0,1\}}
    \Pr(s_1^1)\Pr(s_4^1)
    ~U^{(\bm{s})}(\bm\theta)
    ~\rho~
    U^{(\bm{s})}(\bm\theta)^{\dagger}.
\end{align}
where $\Pr(s^l_j)=(p^l_j)^{s^l_j}(1-p^l_j)^{1-s^l_j}$ with $s^l_j\in\{0,1\}$ as shown in Eq.~\eqref{eq:fig_3_a}. This effective classical control can be incorporated directly into the standard MBQC feedforward procedure. For the unbiased measurement outcomes considered here, a measurement-induced byproduct occurs with probability $1/2$, and deterministic MBQC corrects it through feedforward. The anti-correction scheme instead modifies the same classical control rule to retain (when $s^l_j=1$) or deliberately introduce (when $s^l_j=0$) the corresponding byproduct, allowing the effective byproduct probability to span the full interval $[0,1]$. Since this uses the existing measurement record and feedforward mechanism, it does not increase the average classical-processing overhead relative to deterministic MBQC.

\subsubsection{Shared randomness}
\label{sec:mbqc_realization_of_co_ran}

The previous section described how independent randomness can be generated in MBQC through independently controlled byproducts. We now show how the same feedforward mechanism can realize shared randomness across spatially separated qubits in the cluster-state computation and, equivalently, in the induced circuit model.

Consider the cluster-state computation with $D$ computational layers, as described in Eq.~\eqref{eq:pure_unit}. After the qubits in layer $l$ are measured, their outcomes are recorded, and these raw measurement outcomes are independent. Classical feedforward can process them so that several qubits share the same effective byproduct bit $s^l\in\{0,1\}$. This shared bit controls the collective Pauli byproduct $
Z_{\mathcal{M}}^{s^l}:= \bigotimes_{j\in \mathcal{M}}Z^{s^l}_j$ which acts simultaneously on a subset $\mathcal{M}\subseteq[1,N]$ of qubits. Thus, when $s^l=1$, a $Z$-byproduct is retained on every qubit in $\mathcal M$, whereas all other measurement-induced byproducts are corrected.
As in Sec.~\ref{sec:ind_ran}, we define the effective control variable $c^l$ with 
\begin{equation}
\Pr(c^l=-1)=p^l,
\qquad
\Pr(c^l=+1)=1-p^l,
\end{equation}
such that the shared bit $s^l$ is sampled according to
\begin{equation}\label{eq:fig_3_b}
s^l\sim\mathrm{Bernoulli}(p^l),
~~
\Pr(s^l=1)=p^l,
~
\Pr(s^l=0)=1-p^l.
\end{equation}

As in Sec.~\ref{sec:ind_ran}, the probability $p^l$ can be treated as a tunable parameter in the interval $[0,1]$.

As an example, consider the cluster-state construction shown in the left panel of Fig.~\ref{fig:mbqc_explain}~(b). We choose $\mathcal{M}_l=\{1,4\}$ and use the same control variable $c^1$ for the boundary qubits $(1,1)$ and $(1,4)$. The resulting effective byproducts are therefore $Z_1^{s^1}$ and $Z_4^{s^1}$, which always appear or disappear together depending on $s^1\in\{0,1\}$.

For the corresponding induced circuit shown in the right panel of Fig.~\ref{fig:mbqc_explain}~(b), the branch unitary is
\begin{align}
    U^{(s^1)}(\bm\theta)
    &=
    \left(
        T_c
        \prod_{j=1}^{4}
        \exp\!\left(-i\theta_j^2 Z_j/2\right)
    \right).\\
    &\left(
        T_c
        \left(Z_{1}^{s^1}\otimes Z_{4}^{s^1}\right)
        \prod_{j=1}^{4}
        \exp\!\left(-i\theta_j^1 Z_j/2\right)
    \right).
\end{align}
and the associated correlated channel is
\begin{align}
\label{eq:correlated_mbqc_channel_example}
    \mathcal{E}_{(\bm{\theta}, p^1)}[\rho]
    =
    \sum_{s^1\in\{0,1\}}
    \Pr(s^1)
    ~U^{(s^1)}(\bm\theta)
    ~\rho~
    U^{(s^1)}(\bm\theta)^{\dagger},
\end{align}
where
\begin{equation}
    \Pr(s^1)
    =
    (p^1)^{s^1}(1-p^1)^{1-s^1}.
\end{equation}
The key difference here, from Eq.~\eqref{eq:independent_mbqc_channel_example}, is
that the two endpoint byproducts are not controlled by independently sampled bits. Instead, the same bit $s^1$ controls both byproducts, thereby producing shared randomness in the form of classically correlated byproducts across the two distant qubits. Thus, the example in Fig.~\ref{fig:mbqc_explain}~(b) realizes the model in Def.~\ref{def:min_channel}.

Moreover, this MBQC construction gives a concrete realization of the general stochastic correlated model in Eq.~\eqref{eq:channel_model_set}. More generally, for $L$ correlated byproduct strings $\{P_{{\mathcal{M}}_l}\}_{l=1}^L$, the shared binary variable $s^l$ controls the insertion of the string
\begin{equation}
    P_{\mathcal{M}_l}^{s^l}=\prod_{j\in \mathcal{M}_l} Z_j^{s^l},
    \qquad
    \mathcal{M}_l\subseteq[1,N],
\end{equation}
at a specified layer $l$ in the induced circuit. Assigning each $P_{\mathcal{M}_l}$ to an intermediate insertion slot $l\in\{1,\cdots,D-1\}$ gives a
generic branch unitary as in Eq.~\eqref{eq:unit_with_byp}. If different
shared binary variables are sampled independently, with
$\bm s=(s^1,\ldots,s^L)\in\{0,1\}^L$, then the branch probability is
\begin{equation}
    \pi_{\bm p}(\bm s)
    =
    \prod_{l=1}^{L}
    (p^l)^{s^l}(1-p^l)^{1-s^l}.
\end{equation}
Substituting these branch unitaries from MBQC along with the branch probabilities into
Eq.~\eqref{eq:channel_model_set} gives the correlated quantum channel, and
Eq.~\eqref{eq:channel_dist_all} gives its induced Born distribution.
Although the variables $\{s^1,\ldots,s^L\}$ controlling different Pauli strings are independent, the Pauli operators within each string are correlated because they are controlled by the same bit $s^l$. 

Note that the stochastic Pauli string is not an additional coherent quantum gate. In MBQC, it arises natively from classical processing of random measurement outcomes. This channel character of the model is therefore generated entirely classically.

\subsubsection{End point Pauli-correction}\label{sec:end_pauli_corr}
Here, we first recall the final Pauli-correction step used in the variational measurement-based quantum computation (VMBQC) channel model $\tilde{\mathcal E}_c(\bm \theta,\bm p)$ introduced in Ref.~\cite{majumder2024variational}.

In MBQC, measurement-induced Pauli byproducts, such as $Z^{s^1_1}_1$ and $Z^{s^1_4}_4$ in Fig.~\ref{fig:mbqc_explain}, can be propagated through the remaining cluster state using stabilizer relations. In the equivalent circuit representation, conjugation by Clifford layers maps Pauli operators to Pauli operators. For example, under open boundary conditions, $ T_cX_iT^\dagger_c = X_{i-1} Z_i X_{i+1},~\mathrm{if}~i\in[2,N-1]$. By contrast, propagating a Pauli operator through a non-Clifford Pauli rotation can change the sign of the rotation angle, for instance, $X^s R_z(\theta)=R_z((-1)^s\theta)X^s$. We refer the readers to Ref.~\cite{poulsen2024measurement} for the detailed propagation rules. After propagation through all subsequent layers, the byproducts accumulate into a final Pauli string. For example, for the independent-byproduct construction in Fig.~\ref{fig:mbqc_explain}~(a), the final accumulated byproduct is $\hat P(\bm s)=Z^{s^1_1}_1\otimes X^{s^1_1}_2\otimes X^{s^1_4}_3\otimes Z^{s^1_4}_4$. This final Pauli string can be removed by appending an outcome-dependent correction layer $\hat{P}_{D+1}(\bm{\tilde s})=\hat P(\bm s)^\dagger$ that cancels $\hat P(\bm s)$ before the final measurement. For computational-basis measurements, this is equivalent to classical post-processing of the output bit string. This removes the propagated Pauli byproduct while retaining the branch-dependent intermediate angle flips generated during its propagation. 

A similar endpoint-correction procedure applies to the parametrized circuits presented in Eq.~\eqref{eq:num_channel_main} and shown in Fig.~\ref{fig:learning_model_main}, which consists of parametrized single- and two-qubit Pauli rotations, and more generally to circuits containing parametrized multi-qubit Pauli rotations, as detailed in Sec.~\ref{def:input_byp_model} (see Fig.~\ref{fig:two_learning_models}) of the Appendix.

This distinction is essential for the learning models below. In computational-basis measurements, the effective Pauli operator at the end of the circuit transforms the circuit output trivially: the final $X$ and $Y$ components of the Pauli string merely relabel the output bit string, while the final $Z$ components leave the output probabilities unchanged. By contrast, the intermediate angle flips alter the non-Clifford rotations within the circuit and can therefore modify the output distribution nontrivially. Following Refs.~\cite{majumder2024variational,majumder2026minimizing}, we remove the final propagated Pauli byproduct and retain only the stochastic angle-flip effect in the learning model.

The quantum circuit shown in Eq.~\eqref{eq:model_num_main} along with the channel model with stochastic Pauli string in Eq.~\eqref{eq:num_channel_main} can be implemented using MBQC, which is illustrated in Appendix~\ref{def:input_byp_model}. This model is also used for the numerical experiments in Sec.~\ref{sec:num_results}. The analytical results, however, do not depend on this specific architecture and extend to other finite-range local circuit architectures.

\subsection{Loss and gradients}

For training, we employ the squared maximum mean discrepancy (MMD) \cite{gretton2012kernel} as an implicit loss function, defined below:

\begin{equation}\label{eq:mmd loss}
\begin{aligned}
    \mathcal{L}(\bm{\theta}, \bm{p})= {} & \mathop{\mathbb{E}}_{\substack{x \sim P_{\mathcal{E}_{(\bm{\theta}, \bm{p})}}\\ y \sim P_{\mathcal{E}_{(\bm{\theta}, \bm{p})}}}} [K(x,y)]-2\mathop{\mathbb{E}}_{\substack{x \sim P_{\mathcal{E}_{(\bm{\theta}, \bm{p})}}\\ y \sim Y}} [K(x,y)] \\
    & +\mathop{\mathbb{E}}_{\substack{x \sim Y\\ y \sim Y}} [K(x,y)]
\end{aligned}
\end{equation}
where all samples appearing in the same expectation are drawn independently. Here, $P_{\mathcal{E}_{(\bm{\theta},\bm{p})}}$ is the output distribution of the channel model defined in Eq.~\eqref{eq:channel_dist_all}, and $Y(x)$ is the target distribution generated from the family in Def.~\ref{def:branch_peaked_mixture}. The $K(x,y)$ in Eq.~\eqref{eq:mmd loss} is referred to as the kernel function that quantifies the similarity between samples $x$ and $y$. In practice, the expectations in Eq.~\eqref{eq:mmd loss} are estimated using finite samples from the model and target distributions. When $\bm p=\bm 0$, no stochastic Pauli string is applied and the channel model reduces to the corresponding unitary model, as stated in Observation~\ref{obs:channel_extension}.

During training, the circuit parameters $\bm\theta$ and the shared application probabilities $\bm p$ are updated iteratively using the gradients
\begin{equation}
\frac{\partial\mathcal L}{\partial\theta},
\qquad \rm and \qquad
\frac{\partial\mathcal L}{\partial p},
\end{equation}
for each $\theta\in\bm\theta$ and $p\in\bm p$. In the single Pauli-string model, this includes the derivatives with respect to the circuit parameters $\theta_i^{k,x}$, $\theta_i^{k,z}$, and $\theta_{i,i+1}^{k,zz}$, together with the shared probability $p^l$. These gradients can be evaluated analytically using parameter-shift and stochastic-channel differentiation rules~\cite{liu2018differentiable,majumder2024variational}. Further details are provided in Appendix~\ref{app:grad_details}.

\section{Results}

\subsection{Analytical results}\label{sec:analytical_res}

Shallow-depth local unitary circuits have limited ability to correlate distant
outputs because separated regions depend only on disjoint past light cones.
By contrast, the channel model with shared randomness can use a single classical random variable to apply a spatially extended Pauli string that can generate long-range output correlations without increasing the circuit depth or introducing long-range quantum gates.

\begin{observation}[weak inclusion of model families]
\label{obs:channel_extension}
Consider a fixed depth-$D$ parameterized circuit $U(\bm \theta)$ as in Eq.~\eqref{eq:bw_cir}, with an input state $\rho_{\rm in}$ and computational-basis measurements. Let
\begin{equation}\label{eq:main_unit_dist_fam}
    \mathcal P_{\rm unit}(D)
    :=
    \{P_{U(\bm\theta)}\}_{\bm\theta}
\end{equation}
denote the family of output distributions generated by the unitary model, and let
\begin{equation}\label{eq:main_ch_dist_fam}
    \mathcal P_{\rm ch}(D)
    :=
    \{P_{\mathcal E_{\bm\theta,\bm p}}\}_{\bm\theta,\bm p}
\end{equation}
denote the family of distributions generated by the corresponding correlated
channel model $\mathcal E_{\bm\theta,\bm p}(\rho_{\rm in})$ acting on input state $\rho_{\rm in}$ as defined in Eq.~\eqref{eq:channel_model_set}.

If the parameter domain includes $\bm p=\bm 0$, then
\begin{equation}
    \mathcal P_{\rm unit}(D)
    \subseteq
    \mathcal P_{\rm ch}(D).
\end{equation}
\end{observation}

The previous observation establishes only a weak expressivity statement, i.e., adding the stochastic byproduct cannot reduce the reachable family of distributions. 
The central question is whether the inclusion is strict at a fixed shallow depth $D$, namely whether
\begin{equation}
    \mathcal P_{\rm unit}(D)
    \subsetneq
    \mathcal P_{\rm ch}(D),
\end{equation}
 Theorem~\ref{thm:informal_constant_depth_separation} answers this affirmatively for one-dimensional nearest-neighbor brickwall architectures.

\begin{definition}[Interaction-hypergraph distance]\label{def:graph_dist}
Let $G=(V,E)$ denote the interaction hypergraph of the circuit architecture under consideration (e.g., the local circuit in Eq.~\eqref{eq:bw_cir}), where each vertex in $V$ represents a qubit and each hyperedge $e\in E$ is the support of an allowed local gate. Two qubits are adjacent if they belong to a common hyperedge. For $a,b\in V$, the distance $\operatorname{dist}_{G}(a,b)$ is the minimum number of hyperedges in a sequence connecting $a$ to $b$.
\end{definition}

For a one-dimensional nearest-neighbor architecture with open boundary conditions, the interaction graph is the path graph on $N$ qubits, with edges connecting qubits ($j$) and ($j+1$). The interaction-graph distance is therefore $\operatorname{dist}_G(a,b)=|a-b|.$ In particular, the two boundary qubits $1$ and $N$ are separated by
$\operatorname{dist}_G(1,N)=N-1$. For brevity we write $\operatorname{dist}(a,b):=\operatorname{dist}_G(a,b)$.

\begin{theorem}[\textit{Informal}: shallow-depth channel--unitary separation]
\label{thm:informal_constant_depth_separation}

Consider a one-dimensional nearest-neighbor brickwall circuit architecture on
$N\geq 6$ qubits, together with the minimal correlated channel model of
Def.~\ref{def:min_channel}. In this model a single stochastic Pauli string
$P_{\mathcal{M}}^s$, $s\in\{0,1\}$, acts on two distant qubits $a,b\in M\subseteq [1,N]$ that are separated
by distance $\dist(a,b)$ (see Def.~\ref{def:graph_dist}), at an intermediate layer in the circuit.

Also consider $\mathcal{P}_{unit}(D)$ and $\mathcal P_{\rm ch}(D)$ defined in Eq.~\eqref{eq:main_unit_dist_fam}, and~\eqref{eq:main_ch_dist_fam}, respectively, as the family of distributions corresponding to the unitary and channel model.

Then, there exist a shallow channel model $\mathcal E^* \equiv \mathcal E_{(\bm \theta^*, p^*)}$ with output distribution $Q^\star:= P_{\mathcal E(\bm \theta^\star, p^\star)}$ in the computational basis such that for all depths $1 < D_0 < \frac{\dist(a,b)}{4}$,
%For any fixed depth $1<D_0<\frac{\operatorname{dist}(a,b)}{4}$, there exists channel parameters $(\bm \theta^\star, p^\star)$ such that
\begin{equation}
    % Q^\star
    % :=
    % P_{\mathcal E(\bm \theta^\star, p^\star)}
    % \in
    % \mathcal P_{\rm ch}(D_0), ~~\rm but ~ 
    Q^*\notin \mathcal{P}_{unit}(D_0)
\end{equation}
Moreover, for a constant $\delta^\star>0$ in the number of qubits $N$, and all such depths $D_0$, one can show
\begin{equation}
    \inf_{\theta}
    \TV\!\left(
        Q^\star,
        P_{U(\theta)}
    \right)
    \geq
    \delta^\star
\end{equation} 
where $\TV(\cdot)$ is the total variation distance. In particular, for $\operatorname{dist}(a,b)=\Theta(N)$, the corresponding unitary model requires depth $D_0=\Omega(N)$ to represent $Q^*$.
\end{theorem}
Therefore, the total variation distance between $Q^\star$ and any distribution $P_{U(\bm\theta)}$ generated by the nearest-neighbor unitary circuit is lower bounded by some finite non-zero number $\delta^\star$ for every unitary circuit depth $D_0 < \frac{\operatorname{dist}(a,b)}{4}$. In particular, with $a=1$ and $b=N$, this gives a linear-depth requirement $D=\frac{N-1}{4}$ for nearest-neighbor circuit. The formal statement and proof are given in Appendix~\ref{app:proof_thm_constant_depth_separation} along with the specific expression for $\delta^*$. 

This shallow-depth separation does not rely on Pauli operators acting directly on the final measurement outcomes which can be viewed as trivial classical post-processing. Instead, in the learning models considered here, the final byproducts are corrected or absorbed into the readout (see Sec.~\ref{sec:end_pauli_corr}), while the shared randomness remains encoded in correlated intermediate non-Clifford angle-flips that alter the output distribution non-trivially.

The separation in Theorem~\ref {thm:informal_constant_depth_separation} is not specific to the one-dimensional nearest-neighbor brickwall ansatz. 
\begin{corollary}[Separation for finite-range circuits]
\label{cor:finite_range_architecture_separation}
For any finite-range local circuit architecture, correlated stochastic Pauli operations can generate shallow-depth distributions that are not accessible to the corresponding purely unitary setup at shallow depth.
\end{corollary}

Thus the advantage persists in architectures where long-range entangling gates are not directly available which is formally stated in Corollary~\ref {cor:finite_range_architecture_separation}.

Furthermore, the channel model with shared-randomness used in Theorem~\ref{thm:informal_constant_depth_separation} can be realized natively in MBQC as shown in Sec.~\ref{sec:mbqc_realization_of_co_ran}.
Hence, similar separation holds for the associated MBQC-based generative model.

\begin{remark}[MBQC realization of the separation] \label{cor:mbqc_realization}  As a concrete instantiation, Appendix~\ref{app:mbqc_cluster_proof} gives an explicit MBQC on a cluster-state construction, including the implemented gates, the byproduct propagation, and the verification of the shared randomness conditions used in Theorem~\ref{thm:informal_constant_depth_separation}. \end{remark}

\subsection{Class of distributions}

The separation theorem in the previous section shows that shared randomness can generate branch-dependent output structure in the channel model that is inaccessible to the corresponding shallow unitary model. In this section, we define and analyze a family of distributions that is particularly suitable for these learning models. We therefore consider finite mixtures of peaked distributions, where each component in the mixture is concentrated around a distinct dominant bit string. This structure is naturally matched to the channel model, since different stochastic branches in the channel model can concentrate output probabilities around different dominant bit strings. When the dominant strings differ coherently across distant qubits, the corresponding distribution may contain long-range correlations that the corresponding shallow local unitary model cannot reproduce.

\begin{definition}[Branch-peaked mixture distributions]
\label{def:branch_peaked_mixture}
A distribution $P$ on $\{0,1\}^N$ is called a branch-peaked mixture
distribution if
\begin{equation}\label{eq:branch_peak_class}
    P(x)
    =
    \sum_{\bm s\in\mathcal S}
    \pi(\bm s)
    \left[
        \delta_{\bm s}\mathbf 1[x=x_{\bm s}^\star]
        +(1-\delta_{\bm s})R_{\bm s}(x)
    \right],
\end{equation}
where $\mathcal S\subseteq\{0,1\}^L$ is a finite set satisfying
$|\mathcal S|\leq S_{\max}$, where $S_{\max}=O(1)$ is independent of $N$,
$\pi(\bm s)\geq 0$ and
$\sum_{\bm s\in\mathcal S}\pi(\bm s)=1$.
For each branch $\bm s$, the point
$x_{\bm s}^\star\in\{0,1\}^N$ refers to the peak or dominant bit string of the corresponding branch
distribution, $\delta_{\bm s}\in(0,1]$, and $R_{\bm s}$ is a normalized
residual distribution satisfying $R_{\bm s}(x_{\bm s}^\star)=0$ and
\begin{equation}
    \delta_{\bm s}
    \geq
    (1-\delta_{\bm s})
    \max_{x\neq x_{\bm s}^\star}R_{\bm s}(x).
\end{equation}
Here, $\mathbf 1[\cdot]$ denotes the indicator function.
\end{definition}

Branch-peaked mixtures allow different branches to have branch-dependent peaks and branch-dependent residual distributions. Further details are provided in Appendix~\ref{app:dist_classes}. In Sec.~\ref{sec:num_results}, we further show numerically that suitably designed correlated channel models learn representative distributions from this class more accurately than their unitary competitor.

\subsection{Numerical results}\label{sec:num_results}

The learning models considered in this work include a purely unitary model (Eq.~\eqref{eq:bw_cir}) and a correlated channel model defined in Eq.~\eqref{eq:channel_model_set}. We numerically compare the learning performances of these two models keeping the same width $N=6$ and depth $D=2$.

For the numerical experiment, we consider $N=6$ and construct a
two-branch target distribution,
\begin{equation}\label{eq:num_example_explicit}
    Y(x)
    =
    0.4\,Y_0(x)+0.6\,Y_1(x),
\end{equation}
where
\begin{equation}
    Y_0(x)
    =
    (1-\epsilon_0)^{6-|x|}\epsilon_0^{|x|},
    \qquad
    Y_1(x)
    =
    \epsilon_1^{6-|x|}(1-\epsilon_1)^{|x|},
\end{equation}
with $\epsilon_0=0.15$ and $\epsilon_1=0.25$. Here, $|x|$ denotes the
Hamming weight of $x$. The first branch is peaked at
$x_0^\star=000000$, while the second is peaked at
$x_1^\star=111111$. In the notation of
Definition~\ref{def:branch_peaked_mixture}, the corresponding peak
weights are
$\delta_0=(1-\epsilon_0)^6$ and
$\delta_1=(1-\epsilon_1)^6$, and
\begin{equation}
    R_s(x)
    =
    \begin{cases}
        \dfrac{Y_s(x)}{1-\delta_s},
        & x\neq x_s^\star,\\[6pt]
        0,
        & x=x_s^\star.
    \end{cases}
\end{equation}
Thus,
$Y_s(x)=\delta_s\mathbf{1}[x=x_s^\star]
+(1-\delta_s)R_s(x)$. So the target distribution $Y(x)$ belongs exactly to
the branch-peaked mixture class in Eq.~\eqref{eq:branch_peak_class} with $L=1$ and $\mathcal{S}=\{0,1\}$. Since $\epsilon_0,\epsilon_1<1/2$, the first branch has a unique peak at
$x_0^\star=000000$, while the second has a unique peak at
$x_1^\star=111111$.

We first train the unitary model $U(\bm\theta)$ defined in Eqs.~\eqref{eq:bw_cir} and~\eqref{eq:model_num_main}. Figure~\ref{fig:results_1}~(a) shows the distribution of the minimum MMD loss attained over $300$ training epochs across $20$ independent random initializations. Each value entering the box plot is therefore the best loss obtained in one training run.

We then train the correlated channel model $\mathcal{E}_{\bm\theta,\bm p}$ defined in Eq.~\eqref{eq:num_channel_main} and shown in Fig.~\ref{fig:learning_model_main}. The unitary circuit is now augmented, in the first layer, by the stochastic Pauli string $\bigotimes_{i=1}^6 Z^{s^1}_i$ acting on all qubits immediately after the sublayer $\bigotimes_i R_z(\theta_i^{1,z})$. Since $\bigotimes_{i=1}^6 Z^{s^1}_i$ commutes with the preceding $R_z$ and $R_{zz}$ sub-layers, this placement is equivalent to applying $\bigotimes_{i=1}^6 Z^{s^1}_i$ before the entire first layer. Any circuit layers preceding this insertion slot are set to the identity in the numerical implementation. This Pauli string is controlled by a shared binary variable $s^1\sim\operatorname{Bernoulli}(p^1)$, with $\Pr(s^1=1)=p^1$. For comparison, this corresponds to the circuit shown on the right of Fig.~\ref{fig:model_num}, with the endpoint string $Z_1^{s^1}\otimes Z_N^{s^1}$ replaced by the global string $\bigotimes_{i=1}^6 Z^{s^1}_i$. The trainable parameters are the circuit angles $\bm\theta$, in Eq.~\eqref{eq:layer_params_main}, together with the branch probability $\bm p=\{p^1\}$. The box plot in Fig.~\ref{fig:results_1}~(b) reports the minimum MMD loss attained in each of the $20$ independent random initializations. The correlated channel consistently achieves substantially lower losses than the unitary model and exhibits less variation across random initializations.

For every run, the circuit parameters $\bm\theta$ are initialized independently and uniformly from $[0,1]$, while $p^1$ is initialized uniformly from $[0.5,0.7]$. We use $6000$ samples and a fixed learning rate of $0.5$ for the variational angles $\bm \theta$ and $0.01$ for $p^1$ within the Adagrad optimizer. Further details on the box-plot statistics are provided in Appendix~\ref{app:box_details}.

These finite-size experiments show that,
for the selected target and training protocol, the model containing the matching shared
branch structure reaches lower final error than the unitary baselines. The analytical separation, rather than the numerical experiment, establishes
the representational result.

\begin{figure}[t]
\centering
\includegraphics[width=0.45\textwidth]{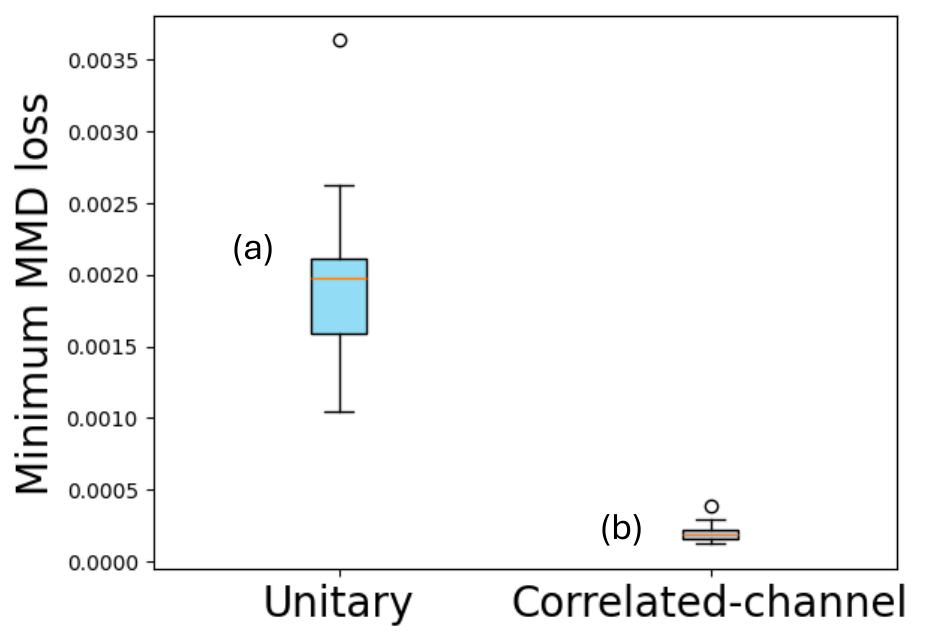}
\caption{\textbf{ Learning performance of MBQC-based model:} Here we examine how well the unitary model $U(\bm \theta)$ and the correlated channel model $\mathcal{E}_{\bm \theta, \bm p}$ can learn a distribution belonging to the family of distributions defined in Def.~\ref{def:branch_peaked_mixture}. The $x$-axis labels the two learning models, while the $y$-axis reports the distribution of minimum MMD losses achieved by $20$ independently initialized unitary and channel learning models over $300$ epochs. Both learning models have $N=6$ qubits and depth $D=2$. (a) Distribution of minimum MMD losses achieved by the unitary model. (b) Distribution of minimum MMD losses achieved by the correlated channel model where we apply a correlated $\bigotimes_{i=1}^6 Z^{s^1}_i$ string on all qubits, in the first layer, with $s^1\in\{0,1\}$ being a binary variable, and probability $\Pr(s^1=1)=p^1\in [0,1]$. For more details regarding the box plot, see Appendix~\ref{app:box_details}.}
\label{fig:results_1}
\end{figure}

\section{Conclusion}

In this work, we demonstrated both analytically and numerically that shared classical randomness can strictly enlarge the set of output distributions accessible to shallow local quantum generative models. The separation does not arise from increasing the quantum depth, adding ancillas, or using long-range entangling gates. Instead, a single stochastic Pauli string, supported on spatially separated qubits and inserted at an intermediate circuit layer, is sufficient to create a channel model whose output distribution cannot be reproduced by the corresponding shallow local unitary Born model. The obstruction is architectural, as in shallow bounded-connectivity circuits, separated output regions can remain outside each other's backward light cones, while the shared random Pauli string can correlate their classical outputs without increasing the quantum circuit depth. For one-dimensional nearest-neighbour architectures, reproducing such distributions, obtained from the channel model, with a purely unitary model requires depth growing linearly with the distance between the correlated regions and can therefore require depth $\Omega(N)$ in the worst case.

Note that unitary circuits with linear depth typically suffer from expressivity-induced barren plateaus, whereas shallow circuits do not~\cite{larocca2025barren, mcclean2018barren, ragone2024lie}. Consistent with this, we observed no significant optimization difficulties in our numerical experiments in Sec.~\ref{sec:num_results}. Moreover, extending the shallow unitary model with shared randomness requires only a few additional trainable parameters. In most of our constructions, a single probability parameter controls the stochastic application of a fixed Pauli string. 

More generally, the same insights will likely apply to any bounded-degree hardware graph, with the required unitary depth controlled by the graph distance between the correlated regions. Thus the result is not specific to a particular gate set or to a strictly one-dimensional geometry. Instead, it is a statement about finite-speed information propagation in local quantum architectures.

The implication for hardware is that classically coordinated stochastic Pauli operations can be a cheap and useful resource for generative modeling on devices where quantum connectivity is limited. On nearest-neighbor superconducting layouts or neutral-atom arrays with local gates, shared randomness can create distributional correlations that would otherwise require deeper routing or longer quantum evolution. In measurement-based implementations, the resource is even more natural as measurement outcomes are already random, and suitable adaptation or retention of selected outcomes realizes the correlated Pauli operations directly. On hardware with native long-range gates, including architectures with high connectivity, the separation based on shared randomness becomes weaker because the relevant graph distances are shorter. Nevertheless, for generative modeling, shared classical randomness can still reduce circuit depth by replacing some coherent long-range operations in the output distribution. Furthermore, on hardware that natively supports the required parametrized Pauli rotations~\cite{debnath2016demonstration,McKay2017eff,kandala2017hardware}, the correlated stochastic Pauli gates can be absorbed into coordinated shifts of existing rotation parameters across the selected qubits, requiring only shared classical control and no additional quantum-gate depth.

These results should not be interpreted as a quantum advantage claim over classical computation. Instead, it is a separation between two quantum model classes under the same shallow local quantum resources. The additional resource is classically cheap: a shared random variable controlling local Pauli operations simultaneously. This shows that stochasticity is not merely a way to include unitary models as special cases, but rather, when structured as shared randomness, it can provide a genuine representational advantage at shallow depth.

Several directions remain open. One is to characterize more completely which target distributions benefit from shared-randomness beyond the branch-peaked mixture distribution used here. In particular, a classification and optimization of different choices of shared randomness for quantum generative modeling is desirable. Another is to optimize the placement and support of the stochastic Pauli strings for realistic hardware graphs. A third is to study trainability: shared randomness enlarges expressivity, but its practical value depends on whether the tunable probabilities corresponding to shared randomness and circuit parameters can be trained efficiently for a particular learning task. Finally, the MBQC realization suggests a broader connection between measurement-induced randomness, adaptive feedforward, and stochastic quantum generative modeling, which may lead to new shallow learning architectures based on classical coordination rather than deeper quantum circuits.

\section{Acknowledgments}
We thank Sofiene Jerbi, Vedran Dunjko, and Damian Markham for useful discussions. This research was funded in part by the Austrian Science
Fund (FWF) [SFB BeyondC F7102, DOI: 10.55776/F71;
WIT9503323, DOI: 10.55776/WIT9503323]. For open access purposes, the author has applied a CC BY public copy-right license to any author accepted manuscript version
arising from this submission. This work was also supported
by the European Union (ERC Advanced Grant, QuantAI,
No. 101055129). The views and opinions expressed in
this article are however those of the author(s) only and do
not necessarily reflect those of the European Union or the
European Research Council - neither the European Union
nor the granting authority can be held responsible for them.

GPT-5.6 Sol was used to help establish the consistency of the mathematical statements and proofs during the preparation of the manuscript. All scientific ideas, analyses, and conclusions are those of the authors, and all AI-improved text was carefully reviewed, manually checked, and validated prior to inclusion.
\section{Data availability}
The code and numerical data supporting the findings of this study are available at \cite{Majumder_Shared_randomness}.
\bibliography{ref}% Produces the bibliography via BibTeX.

\newpage
\onecolumngrid
\section{Related work}
\label{sec:related_work}

Generically, previous works adding classical randomness into parameterized quantum circuits mostly rely on the observation that quantum channels are more general than unitaries. However, we are not aware of other works systematically investigating the crucial benefits of shared classical randomness for generative modeling with shallow quantum circuits. Furthermore, we focus on keeping the number of additional trainable parameters small. Inspired by MBQC, we often have just few additional trainable probabilities deciding whether a few fixed Pauli strings get placed or not at an intermediate layer in the circuit. This requires no long-range entangling gates, no mid-circuit measurements, and no controlled arbitrary unitaries. Despite this minimal modification, the resulting channel model can generate long-range output correlations that shallow local unitary models cannot represent, because the relevant backward light cones of the latter remain disjoint. Thus, our contribution is not the generic claim that ``randomness helps'', nor an empirical optimization effect, but strict complexity-theoretical shallow-depth separations for quantum generative models enabled by shared randomness.

Prior works have demonstrated that stochasticity, mixtures, and non-unitary extensions of variational quantum circuits can improve expressivity or trainability compared with purely unitary architectures. Wu et al.~\cite{wu2024randomness} introduced randomness-enhanced quantum neural networks, where a parameterized unitary circuit is enhanced by a random unitary layer sampled from an ensemble. These randomly sampled unitaries consist of trainable single qubit rotations. Instead of such trainable random unitaries, we chose to sample only fixed Pauli strings to keep the parameter count smaller. Using majorization theory of observables, they analytically and also numerically show that stochasticity can increase the expressive power of quantum models for supervised learning tasks. Different from our work, the analytical results take advantage of restrictions given by unitarity itself (unitary channels cannot change eigenvalues of observables), rather than restrictions introduced by shallow ansätze of local connectivity.

Similarly, Coyle et al.~\cite{coyle2025training} proposed density quantum neural networks, in which the model outputs an average over a weighted ensemble of trainable unitary branches. Equivalently, each circuit execution samples one branch from the ensemble, runs the corresponding parameterized sub-unitary, and estimates observables by averaging over branches. This results in a density matrix similar to Local-Combination-of-Unitaries (LCU), but with smaller resource overhead. However, also here, each branch may still be a nontrivial trainable circuit with its own parameters, depth, and entangling structure. The formal results analyze the training cost and in particular the number of samples, and compare the resource costs to the standard LCU scheme. The interplay of shallowness and shared randomness is not explicitly explored.  

The LCU scheme itself has been analyzed in Heredge et al.~\cite{heredge2025nonunitary} for adapting classical machine learning principles such as residual connections and pooling layers. Such constructions can increase expressivity, but they typically require coherent control over multiple unitary components, which can substantially increase circuit depth and implementation overhead. Furthermore, the non-unitary is not achieved by the injection of randomness, but by post-selection.

Wen et al.~\cite{wen2026trainable} formulate trainable quantum channels as computational primitives, emphasizing that generically, channels and their parameterized Kraus operators strictly generalize unitary models. Also here, the argument that unitaries cannot change eigenvalues of observables is used. They also argue that the loss landscape is not constrained to a unitary manifold anymore. As channel ans\"{a}tze, they investigate unitary circuits supplemented with tunable amplitude- and phase-damping channels, which they interpret as tunable dissipation.

Noise and dissipation have also been used productively in quantum learning. In quantum reservoir computing, Domingo et al.~\cite{domingo2023taking}, Fry et al.~\cite{fry2023noise_induced_qrc}, Sannia et al.~\cite{sannia2024dissipation_qrc}, and Kubota et al.~\cite{kubota2023temporal_noise} show that noise or dissipation can improve temporal information processing or reservoir performance. These works mostly consider single qubit noise, although the last reference also considers noise affecting two nearby qubits. These works acknowledge the usefulness of noise or dissipation for reducing quantum resources, but they do not investigate the benefits of shared randomness over distant qubits. Their noise is also fixed, while we consider trainable shared randomness.

Other works use stochastic noise for regularization, robustness, or optimization: Scala et al.~\cite{scala2025dropout} study quantum dropout, Kuzmin et al.~\cite{kuzmin2025noise_regularization} study noise-induced regularization, Huang et al.~\cite{huang2023certified_robustness} and Winderl et al.~\cite{winderl2024optimal_noise} use noise for adversarial robustness, and Liu et al.~\cite{liu2025stochastic_noise_vqa} show that stochastic noise can help variational algorithms escape strict saddles. These results show that noise need not be detrimental, but none identify shared classical randomness as a resource for overcoming locality constraints in shallow generative Born models.

Variational measurement-based quantum computation (VMBQC) was introduced as a framework in which
measurement-induced Pauli byproducts are retained with trainable probabilities and exploited for quantum generative modeling
~\cite{majumder2024variational}. More recently, it was shown that even a
restricted VMBQC channel with a single additional trainable probability
can represent distributions inaccessible to the corresponding unitary
model~\cite{majumder2026minimizing}. The present work extends these
architecture-specific observations by identifying shared classical
randomness as the relevant resource and establishing a strict
shallow-depth separation for general local quantum circuits.

Several other works have shown that one can generate long-range correlations by preparing pure states with long-range entanglement. Dynamic circuits and LOCC-assisted protocols use ancilla qubits, measurements and classical communication to generate long-range entanglement or reduce state-preparation depth~\cite{baumer2024efficient,piroli2024approximating}. Those works show that non-unitary operations and classical communication can overcome geometric limitations of local unitary circuits. Our setting uses the same broad physical principle, but for a different task: learning classical output distributions. However, our scheme needs no ancilla qubits to prepare superposition of  states with long-range entanglement, and no classical communication or feed-forward operations between qubits. Instead, our scheme can be interpreted as using shared classical randomness. 

Distributed QML schemes connect quantum processors using classical channels or real-time communication~\cite{hwang2025distributed}. This motivates classical communication as a scalable architectural primitive. Our model is lower-level and more specific: the classical channel samples a shared random variable that coordinates Pauli operations across qubits, thereby modifying the generated probability distribution. Our scheme relies on shared classical randomness rather than measurement-conditioned classical communication. 

More generally, quantum generative models aim to exploit structured correlations that are difficult to capture with restricted classical models. For example, Gao et al.~\cite{gao2022enhancing} showed that quantum correlations can enhance generative modeling by comparing classical Bayesian-network models with basis-enhanced Bayesian quantum circuits. Their separations use quantum non-locality and contextuality to show advantages over classical $k$-gram and hidden-Markov models.

\newpage
\section*{Supplementary Materials}\label{sec:end_matter}

\appendix

\section{Background and definitions}\label{app1:background}

\paragraph{Brickwall circuit architecture.}
Throughout the separation proofs, the local unitary model is a one-dimensional nearest-neighbor brickwall circuit on $N$ qubits with open boundary conditions. It is written as
\begin{equation}\label{eq:app_bw_cir}
U(\bm{\theta})
=
U_D(\bm{\theta}_D)\cdots U_1(\bm{\theta}_1),
\end{equation}
where $D$ denotes the number of brickwall layers. Each layer consists of two sub-layers of nearest-neighbor two-qubit gates,
\begin{equation}\label{eq:app_sub_layers}
U_k(\bm{\theta}_k)
=
U^{(k)}_{\mathrm{odd}}
U^{(k)}_{\mathrm{even}},
\end{equation}
with
\begin{equation}
U^{(k)}_{\mathrm{even}}
=
\prod_{\substack{i=2,~ i\ \mathrm{even}}}^{N-1}
G^{(k)}_{i,i+1},
\qquad
U^{(k)}_{\mathrm{odd}}
=
\prod_{\substack{i=1,~ i\ \mathrm{odd}}}^{N-1}
G^{(k)}_{i,i+1}.
\end{equation}
Here, ($G^{(k)}_{i,i+1}$) is an arbitrary two-qubit gate acting on the neighboring sites $(i,i+1)$. The gates within each sublayer have disjoint supports and therefore can be applied in parallel, whereas the ordering of the even- and odd-bond sublayers is part of the circuit definition. Arbitrary single-qubit gates may be inserted before, after, or between these sublayers without changing the spatial locality of the architecture. Open boundary conditions restrict the two-qubit gates to the edges $(i,i+1)$, with ($i=1,\ldots,N-1$). This architecture defines the local unitary baseline used in the separation results.

\paragraph{Input state.}
We consider a fixed product input state
\begin{equation}
\rho_{\rm in}
=
\bigotimes_{j=1}^N \rho_j,
\end{equation}
where each $\rho_j$ is an arbitrary single-qubit density operator. The state may therefore be pure or mixed and need not be identical across the qubits. The product-state assumption ensures that the input contains no correlations between spatially separated regions.

\paragraph{Stochastic Pauli string.}
Let $M\subseteq\{1,\ldots,N\}$ denote a subset of qubits containing at least two distinct sites. A Pauli string supported on $M$ is defined as
\begin{equation}
P_{\mathcal{M}}
:=
\prod_{j\in {\mathcal{M}}} Q_j,
\qquad
Q_j\in\{X_j,Z_j\},
\end{equation}
with the identity acting on all qubits outside $M$. A single Bernoulli random variable $s\in\{0,1\}$, with
\begin{equation}
\Pr(s=1)=p,
\qquad
\Pr(s=0)=1-p,
\end{equation}
controls the application of the entire string,
\begin{equation}
P_{\mathcal{M}}^s
=
\begin{cases}
\id, & s=0,\\
P_{\mathcal{M}}, & s=1.
\end{cases}
\end{equation}
The same random bit therefore determines the Pauli operations on each site in $M$ simultaneously. Averaging over $s$ gives the stochastic Pauli channel
\begin{equation}
\mathcal{E}_{p}(\rho)
=
(1-p)\rho
+
pP_{\mathcal{M}}\rho P_{\mathcal{M}}^\dagger.
\end{equation}
for some input state $\rho$. The stochastic Pauli string $P_{\mathcal{M}}^s$ may be inserted between any two consecutive layers of the brickwall circuit in Eq.~\eqref{eq:app_bw_cir}.

\paragraph{Unitary and channel output distributions}
For a fixed depth $D$, input state $\rho_{\rm in}$, and local circuit architecture, let $U(\bm{\theta})$ denote the parametrized unitary circuit. After computational-basis measurement, its output distribution, as shown in Eq.~\eqref{eq:unit_dist}, is
\begin{equation}
P_{U(\bm{\theta})}(x)
=
\Tr\left[
\ket{x}\bra{x}
U(\bm{\theta})\rho_{\rm in}U(\bm{\theta})^\dagger
\right],
\qquad
x\in\{0,1\}^N.
\end{equation}
The corresponding family of unitary output distributions is
\begin{equation}\label{eq:unit_dist_family}
\mathcal P_{\rm unit}(D)
:=
\left\{
P_{U(\bm{\theta})}
:
\bm{\theta}\in \bm \Theta_D
\right\}.
\end{equation}
where $\bm \Theta_D$ denotes the parameter domain of the depth-$D$ circuit.

The corresponding correlated channel model, as defined in Eq.~\eqref{eq:channel_model_set}, is
\begin{equation}\label{eq:app_channel_model}
\mathcal E_{\bm\theta,\bm p}(\rho_{\rm in})
    =
    \sum_{\bm s\in\{0,1\}^L}
    \pi_{\bm p}(\bm s)\,
    U^{(\bm s)}(\bm\theta)
    \rho_{\rm in}
    U^{(\bm s)}(\bm\theta)^\dagger 
\end{equation}
where \(U^{(\bm{s})}(\bm{\theta})\) denotes the branch unitary corresponding to a fixed realization
\(\bm{s}=(s^1,\ldots,s^L)\in\{0,1\}^L\)
of the stochastic control variables. For a single stochastic Pauli string, this reduces to
\(s\in\{0,1\}\) and the corresponding branch unitary \(U^{(s)}(\bm{\theta})\). The computational-basis output distribution is
\begin{equation}
P_{\mathcal E_{\bm{\theta},\bm p}}(x)
=
\Tr\left[
\ket{x}\bra{x}
\mathcal E_{\bm{\theta},\bm p}(\rho_{\rm in})
\right].
\end{equation}
The family of channel output distributions is therefore
\begin{equation}\label{eq:ch_dist_family}
\mathcal P_{\rm ch}(D)
:=
\left\{
P_{\mathcal E_{\bm{\theta},\bm p}}
:
\bm{\theta}\in \bm \Theta_D,~~\bm p\in[0,1]^L
\right\}.
\end{equation}
where $L$ is the number of stochastic Pauli strings in the circuit. Here, $\ket{x}=\ket{x_1}\otimes\cdots\otimes\ket{x_N}$ denotes the computational-basis state associated with $x=(x_1,\ldots,x_N)\in \{0,1\}^N$.

\paragraph{Total variation distance.}
Total variation distance is a metric for comparing two output
probability distributions. For distributions $P$ and $Q$ on the same finite sample space $\Omega$, it is defined as
\begin{equation}
    \TV(P,Q)
    :=
    \frac12\sum_{x\in\Omega}|P(x)-Q(x)|.
\end{equation}

Thus, a lower bound on total variation distance gives a strong distributional separation: if $\TV(P,Q)\geq \delta$, then the two models cannot approximate each other closer than error $\delta$ at the level of output distributions.
\begin{definition}[$Z$-covariance]
\label{def:z_covariance}
Let $\rho$ be an $N$-qubit state and let $\rho_{ab}$ denote its reduced state on qubits $a$ and $b$. The $Z$-covariance between these qubits is defined as the connected two-point correlator
\begin{equation}
    \Cov_{\rho_{ab}}(\hat Z_a,\hat Z_b)
    :=
    \langle \hat Z_a\otimes \hat Z_b\rangle_{\rho_{ab}}
    -
    \langle \hat Z_a\rangle_{\rho_{ab}}
    \langle \hat Z_b\rangle_{\rho_{ab}},
\end{equation}
where
\begin{equation}
    \langle \hat Z_a\rangle_{\rho_{ab}}
    :=
    \Tr[\rho_{ab}(\hat Z_a\otimes \id)],
    \qquad
    \langle \hat Z_b\rangle_{\rho_{ab}}
    :=
    \Tr[\rho_{ab}(\id\otimes \hat Z_b)],
    \qquad \langle \hat Z_a\otimes \hat Z_b\rangle_{\rho_{ab}}
    :=
    \Tr[\rho_{ab}(\hat Z_a\otimes \hat Z_b)].
\end{equation}
\end{definition}

\begin{lemma}[Backward light cones in a nearest neighbor brickwall circuit]
\label{lem:light-cone}
Consider an open-chain nearest-neighbour brickwall circuit on $N$ qubits as in Eq.~\eqref{eq:app_bw_cir}, where each layer in the circuit can have arbitrary single-qubit gates and nearest-neighbour two-qubit gates. For any qubit $a\in\{1,\ldots,N\}$, let $L_D(a)$ denote the support of the backward light cone of its final computational-basis measurement after $D$ layers of the circuit. Then
\begin{equation}
L_D(a)
\subseteq
\left\{
\max\{1,a-2D\},
\ldots,
\min\{N,a+2D\}
\right\}.
\end{equation}
Consequently, for any two distinct qubits $a,b\in\{1,\ldots,N\}$, 
\begin{equation}
D
<
\frac{\operatorname{dist}(a,b)}{4},
\qquad
\operatorname{dist}(a,b):=|a-b|,
\end{equation}
implies that their backward light cones are disjoint:
\begin{equation}
L_D(a)\cap L_D(b)=\varnothing.
\end{equation}
In particular, for the boundary qubits $a=1$ and $b=N$,
\begin{equation}
L_D(1)
\subseteq
\{1,\ldots,2D+1\},
\qquad
L_D(N)
\subseteq
\{N-2D,\ldots,N\},
\end{equation}
and these supports are disjoint whenever
\begin{equation}
D<\frac{N-1}{4}.
\end{equation}
\end{lemma}
Similar bounds hold for more general local circuit architectures, where the depth $D$ is determined by the interaction range and geometry of the underlying connectivity graph.
\begin{proof}
Consider two qubits $a,b\in\{1,\ldots,N\}$, and assume without loss of generality that $a<b$. The backward light cone of the final computational-basis measurement on qubit $a$ is obtained by propagating the observable $\hat Z_a$ backwards through the circuit. Similarly, the backward light cone of the measurement on qubit $b$ is obtained by propagating $\hat Z_b$ backwards.

Single-qubit gates do not enlarge the support of an operator. A nearest-neighbor two-qubit gate acting on qubits $(i,i+1)$ can enlarge the support only across the edge $(i,i+1)$. In particular, if the operator is supported on one endpoint of the gate, conjugation may extend its support to the other endpoint, whereas a gate whose support is disjoint from that of the operator leaves the operator support unchanged.

Each circuit layer contains at most two nearest-neighbor two-qubit sublayers (see Eq.~\eqref{eq:app_sub_layers}). Hence, during one full layer, the support of a backward-propagated operator can expand by at most two sites in each direction. Starting from the single-site support $a$, after $D$ layers the backward light cone therefore satisfies
\begin{equation}
L_D(a)
\subseteq
\left\{
\max\{1,a-2D\},
\ldots,
\min\{N,a+2D\}
\right\}.
\end{equation}
The same argument gives
\begin{equation}
L_D(b)
\subseteq
\left\{
\max\{1,b-2D\},
\ldots,
\min\{N,b+2D\}
\right\}.
\end{equation}

Since $a<b$, these two intervals are disjoint whenever $a+2D<b-2D.$ Equivalently, $4D<b-a = \operatorname{dist}(a,b)$ or $D< \frac{\operatorname{dist}(a,b)}{4}$ Therefore,
\begin{equation}
L_D(a)\cap L_D(b)=\varnothing
\end{equation}
under the stated depth condition.
\end{proof}

\begin{lemma}[Two-qubit computational-basis marginal]
\label{lem:two_qubit_z_marginal}
Let $\rho$ be an arbitrary $N$-qubit state, and let $\rho_{ab}$ denote its reduced state on qubits $a$ and $b$. Measuring these two qubits in the computational basis produces the joint distribution $P(x_a,x_b)$, where $x_a,x_b\in\{0,1\}$. Then
\begin{equation}
\label{eq:joint_marginal_1}
    P(x_a,x_b)
    =
    \frac{1}{4}
    \left(
        1
        +
        (-1)^{x_a}\langle \hat Z_a\rangle
        +
        (-1)^{x_b}\langle \hat Z_b\rangle
        +
        (-1)^{x_a+x_b}
        \langle \hat Z_a\otimes \hat Z_b\rangle
    \right),
\end{equation}
where
\begin{equation}
    \langle \hat Z_a\rangle
    :=
    \Tr[\rho_{ab}(\hat Z_a\otimes \id)],
    \qquad
    \langle \hat Z_b\rangle
    :=
    \Tr[\rho_{ab}(\id\otimes \hat Z_b)],
    \qquad
    \langle \hat Z_a\otimes \hat Z_b\rangle
    :=
    \Tr[\rho_{ab}(\hat Z_a\otimes \hat Z_b)].
\end{equation}
\end{lemma}

\begin{proof}
The joint probability of outcomes $x_a,x_b\in\{0,1\}$ is
\begin{equation}
\label{eq:joint_marginal_2}
    P(x_a,x_b)
    =
    \Tr[
        \rho_{ab}
        (\Pi_{x_a}\otimes \Pi_{x_b})
    ],
\end{equation}
where
\begin{equation}
    \Pi_{x_i}
    =
    \ket{x_i}\bra{x_i}
    =
    \frac{1}{2}
    \left(
        \id
        +
        (-1)^{x_i}\hat Z_i
    \right),
    \qquad
    i\in\{a,b\}.
\end{equation}
Therefore,
\begin{equation}
    \Pi_{x_a}\otimes \Pi_{x_b}
    =
    \frac{1}{4}
    \left(
        \id\otimes\id
        +
        (-1)^{x_a}\hat Z_a\otimes\id
        +
        (-1)^{x_b}\id\otimes\hat Z_b
        +
        (-1)^{x_a+x_b}
        \hat Z_a\otimes\hat Z_b
    \right).
\end{equation}
Substituting this expression into Eq.~\eqref{eq:joint_marginal_2} gives exactly Eq.~\eqref{eq:joint_marginal_1},
since $\Tr[\rho_{ab}]=1$.
\end{proof}

\begin{lemma}[Disjoint light cones imply factorized two-site marginal]
\label{lem:disjoint_lightcone_factorized_marginal}
Consider an open chain depth-$D$ one-dimensional nearest-neighbor brickwall unitary circuit
$U(\bm\theta)$ (Eq.~\eqref{eq:app_bw_cir}) acting on an $N$-qubit product input state $\rho_{\rm in}
    =
    \bigotimes_{i=1}^N \rho_i $.
Let $a$ and $b$ be two output sites separated by distance $\dist(a,b)$, and let
\begin{equation}
    \rho_{ab}
    :=
    \Tr_{[1,N]\setminus\{a,b\}}
    \left[
        U(\bm\theta)\rho_{\rm in}U(\bm\theta)^\dagger
    \right]
\end{equation}
be the reduced output state on sites $a$ and $b$. Then $\Cov_{\rho_{ab}}(\hat Z_a,\hat Z_b)=0$ for all depth $D<\frac{\dist(a,b)}{4}$.
Consequently, the computational-basis two-site joint marginal factorizes as
\begin{equation}
    P_{U(\bm\theta)}(x_a,x_b)
    =
    P_{U(\bm\theta)}(x_a)P_{U(\bm\theta)}(x_b),
    \qquad
    x_a,x_b\in\{0,1\}.
\end{equation}
\end{lemma}

\begin{proof}
Define the Heisenberg-evolved observables
\begin{equation}
    \hat O_a
    :=
    U(\bm\theta)^\dagger
    (\hat Z_a\otimes \id_{[1,N]\setminus \{a\}})
    U(\bm\theta),
    \qquad
    \hat O_b
    :=
    U(\bm\theta)^\dagger
    (\hat Z_b\otimes \id_{[1,N]\setminus \{b\}})
    U(\bm\theta).
\end{equation}
From Lemma~\ref{lem:light-cone}, for all depth $D<\frac{\dist(a,b)}{4}$ the backward light cones of $a$ and $b$ are disjoint. Hence
$\hat O_a$ and $\hat O_b$ are supported on disjoint sets of input qubits. Therefore, using the product structure of $\rho_{\rm in}$,

\begin{align}
    \Tr[\rho_{ab}(\hat Z_a\otimes\hat Z_b)]
    &=
    \Tr[\rho_{\rm in}\hat O_a\hat O_b] \\
    &=
    \Tr[\rho_{\rm in}\hat O_a]\Tr[\rho_{\rm in}\hat O_b] \\
    &=
    \Tr[\rho_{ab}(\hat Z_a\otimes\id)]
    \Tr[\rho_{ab}(\id\otimes\hat Z_b)] .
\end{align}

Thus
\begin{align}
    \Cov_{\rho_{ab}}(\hat Z_a,\hat Z_b)
    =
    \Tr[\rho_{ab}(\hat Z_a\otimes\hat Z_b)]
    -
    \Tr[\rho_{ab}(\hat Z_a\otimes\id)]
    \Tr[\rho_{ab}(\id\otimes\hat Z_b)]
    =
    0.
\end{align}

Now apply Lemma~\ref{lem:two_qubit_z_marginal}. Write
\begin{align}
    \alpha
    :=
    \Tr[\rho_{ab}(\hat Z_a\otimes\id)],
    \qquad
    \beta
    :=
    \Tr[\rho_{ab}(\id\otimes\hat Z_b)],
    \qquad
    \gamma
    :=
    \Tr[\rho_{ab}(\hat Z_a\otimes\hat Z_b)].
\end{align}

The zero-covariance condition gives $\gamma=\alpha\beta$. Therefore,

\begin{align}
    P_{U(\bm\theta)}(x_a,x_b)
    &=
    \frac14
    \left(
        1
        +
        (-1)^{x_a}\alpha
        +
        (-1)^{x_b}\beta
        +
        (-1)^{x_a+x_b}\gamma
    \right) \\
    &=
    \frac14
    \left(
        1
        +
        (-1)^{x_a}\alpha
        +
        (-1)^{x_b}\beta
        +
        (-1)^{x_a+x_b}\alpha\beta
    \right) \\
    &=
    \frac12
    \left(
        1+(-1)^{x_a}\alpha
    \right)
    \frac12
    \left(
        1+(-1)^{x_b}\beta
    \right) \\
    &=
    P_{U(\bm\theta)}(x_a)P_{U(\bm\theta)}(x_b).
\end{align}

This proves the factorization of the two-site computational-basis marginal.
\end{proof}

\begin{lemma}[Marginal obstruction implies full-distribution obstruction]
\label{lem:marginal_obstruction_classical}
Let $P$ be a target probability distribution on $\{0,1\}^N$, and let
\begin{equation}
    \mathcal Q
    :=
    \{Q_\theta:\theta\in\Theta\}
\end{equation}
be a model class of probability distributions on $\{0,1\}^N$.

For any subset $M\subseteq [1,N]$, denote the marginal of $P$ on $M$ by $P_{\mathcal{M}}$, and the corresponding marginal of $Q_\theta$ by $Q_{\theta,M}$. 
If there exists a subset $M\subseteq [1,N]$ such that
\begin{equation}
    P_{\mathcal{M}}
    \notin
    \{Q_{\theta,M}:\theta\in\Theta\},\qquad
    \mathrm{then}\qquad
    P
    \notin
    \mathcal Q.
\end{equation}
Equivalently, if no model distribution has the correct marginal on $M$, then no model distribution can equal the full target distribution.
\end{lemma}

\begin{proof}
If $P=Q_\theta$, then marginalizing both sides gives $P_{\mathcal{M}}=Q_{\theta,M}$, contradicting the assumption.
\end{proof}

\begin{lemma}[Quantitative marginal obstruction]
\label{lemma:marginal_obstruction}
Let $d$ be any statistical distance or divergence that is contractive under marginalization, namely
\begin{equation}
    d(P_{\mathcal{M}},Q_{\mathcal{M}})
    \leq
    d(P,Q)
\end{equation}
for all probability distributions $P,Q$ on $\{0,1\}^N$ and all subsets $M\subseteq [1,N]$. Then, for every target distribution $P$ and every model class
$ \mathcal Q
    :=
    \{Q_\theta:\theta\in\Theta\},$
one has
\begin{equation}
    \inf_{\theta\in\Theta}
    d(P,Q_\theta)
    \geq
    \inf_{\theta\in\Theta}
    d(P_{\mathcal{M}},Q_{\theta,M})
\end{equation}
for every $M\subseteq [1,N]$. In particular, if for some $M\subseteq [1,N]$,
\begin{equation}
    \inf_{\theta\in\Theta}
    d(P_{\mathcal{M}},Q_{\theta,M})
    \geq
    \varepsilon
    >
    0,\qquad \mathrm{then} \qquad \inf_{\theta\in\Theta}
    d(P,Q_\theta)
    \geq
    \varepsilon.
\end{equation}
Thus, a nonzero obstruction at the marginal level implies a nonzero obstruction at the full-distribution level.
\end{lemma}

\begin{proof}
The claim follows by contractivity under marginalization:
\begin{equation}
    d(P,Q_\theta)\ge d(P_{\mathcal{M}},Q_{\theta,{\mathcal{M}}})
\end{equation}
for every $\theta$, and then taking the infimum over $\theta$. 
\end{proof}
For all analytical bounds in this paper, we mostly use the total variation (TV) distance as the divergence metric.

\begin{lemma}[Diagonal observable gap implies total-variation gap]
\label{lem:diagonal_observable_tv_gap}
Let $\rho$ and $\sigma$ be two $N$-qubit states, and let $P_\rho$ and $P_\sigma$ be their Born distributions obtained by measuring all qubits in the computational basis. Let
\begin{equation}
    A
    =
    \sum_{x\in\{0,1\}^N} a_x \Pi_x,
    \qquad
    \Pi_x := |x\rangle\langle x|
\end{equation}
be a diagonal observable satisfying $\|A\|_\infty\leq 1$. Then
\begin{equation}
    \left|
        \Tr[A\rho] - \Tr[A\sigma]
    \right|
    \leq
    2\TV(P_\rho,P_\sigma).
\end{equation}
Equivalently,
\begin{equation}
    \TV(P_\rho,P_\sigma)
    \geq
    \frac12
    \left|
        \Tr[A\rho] - \Tr[A\sigma]
    \right|.
\end{equation}
\end{lemma}

\begin{proof}
Since $A$ is diagonal in the measurement basis,

\begin{align}
    \left|
        \Tr[A\rho] - \Tr[A\sigma]
    \right|
    &=
    \left|
        \sum_x a_x\bigl(P_\rho(x)-P_\sigma(x)\bigr)
    \right| \\
    &\leq
    \sum_x |a_x|\,|P_\rho(x)-P_\sigma(x)| \\
    &\leq
    \sum_x |P_\rho(x)-P_\sigma(x)| \\
    &=
    2\TV(P_\rho,P_\sigma).
\end{align}

\end{proof}

\begin{lemma}[Covariance gap implies total-variation gap]
\label{lem:covariance_tv_gap}
Let $\rho$ and $\sigma$ be two arbitrary $N$-qubit states, and let $\rho_{ab}$ and $\sigma_{ab}$ denote their respective reduced states on qubits $a$ and $b$. Let $P_{\rho_{ab}}$ and $P_{\sigma_{ab}}$ be the corresponding joint distributions obtained by measuring qubits $a$ and $b$ in the computational basis. Define
\begin{align}
    \Cov_{\rho}(\hat Z_a,\hat Z_b)=\Cov_{\rho_{ab}}(\hat Z_a,\hat Z_b)
    :=
    \Tr[\rho_{ab}(\hat Z_a\otimes \hat Z_b)]
    -
    \Tr[\rho_{ab}(\hat Z_a\otimes \id)]
    \Tr[\rho_{ab}(\id\otimes \hat Z_b)] ,
\end{align}
and similarly for $\sigma_{ab}$. Then
\begin{align}
    \TV(P_{\rho_{ab}},P_{\sigma_{ab}})
    \geq
    \frac{1}{6}
    \left|
        \Cov_{\rho_{ab}}(\hat Z_a,\hat Z_b)
        -
        \Cov_{\sigma_{ab}}(\hat Z_a,\hat Z_b)
    \right|.
\end{align}
\end{lemma}

\begin{proof}
Set
\begin{align}
    \alpha_\rho:=\Tr[\rho_{ab}(\hat Z_a\otimes\id)],
    \qquad
    \beta_\rho:=\Tr[\rho_{ab}(\id\otimes\hat Z_b)],
    \qquad
    \gamma_\rho:=\Tr[\rho_{ab}(\hat Z_a\otimes\hat Z_b)],
\end{align}
and define $\alpha_\sigma,\beta_\sigma,\gamma_\sigma$ analogously. By
Lemma~\ref{lem:diagonal_observable_tv_gap}, applied to
$\hat Z_a\otimes\id$, $\id\otimes\hat Z_b$, and
$\hat Z_a\otimes\hat Z_b$, we have
\begin{align}
    |\alpha_\rho-\alpha_\sigma|
    \leq 2\TV(P_{\rho_{ab}},P_{\sigma_{ab}}),
\qquad
    |\beta_\rho-\beta_\sigma|
    \leq 2\TV(P_{\rho_{ab}},P_{\sigma_{ab}}),
\qquad
    |\gamma_\rho-\gamma_\sigma|
    \leq 2\TV(P_{\rho_{ab}},P_{\sigma_{ab}}).
\end{align}
Since $|\alpha_\rho|,|\beta_\sigma|\leq 1$, we get

\begin{align}
&
\left|
    \Cov_{\rho_{ab}}(\hat Z_a,\hat Z_b)
    -
    \Cov_{\sigma_{ab}}(\hat Z_a,\hat Z_b)
\right|
\\
&=
\left|
    \gamma_\rho-\alpha_\rho\beta_\rho
    -
    \gamma_\sigma+\alpha_\sigma\beta_\sigma
\right|
\\
&\leq
|\gamma_\rho-\gamma_\sigma|
+
|\alpha_\rho\beta_\rho-\alpha_\sigma\beta_\sigma|
\\
&\leq
|\gamma_\rho-\gamma_\sigma|
+
|\alpha_\rho|\,|\beta_\rho-\beta_\sigma|
+
|\beta_\sigma|\,|\alpha_\rho-\alpha_\sigma|
\\
&\leq
6\TV(P_{\rho_{ab}},P_{\sigma_{ab}}).
\end{align}

Rearranging proves the claim.
\end{proof}
\section{Proof of Theorem \ref{thm:informal_constant_depth_separation}}\label{app:proof_thm_constant_depth_separation}
\begin{theorem}[\textit{Formal}: shallow-depth channel--unitary separation]
\label{thm:constant_depth_channel_separation}

Consider a one-dimensional nearest-neighbor brickwall circuit architecture on
$N\geq 6$ qubits, together with the minimal correlated channel model of
Def.~\ref{def:min_channel}, in which a single stochastic Pauli string
$P_{\mathcal{M}}^s=P^s_a\otimes P^s_b$, $s\in\{0,1\}$, acts on two distant qubits $a,b\in \mathcal{M}$ that are separated
by distance $\dist(a,b)$ (see Def.~\ref{def:graph_dist}). Also consider circuit depths $1<D< \frac{\operatorname{dist}(a,b)}{4}$. Let
$\mathcal P_{\rm unit}(D)$ and $\mathcal P_{\rm ch}(D)$ denote the family of output distributions defined in Eqs.~\eqref{eq:unit_dist_family} and \eqref{eq:ch_dist_family}, respectively.

Let $\hat Z_a$ and $\hat Z_b$ denote the final local Pauli observables 
at output sites $a$ and
$b$. Suppose that there exists a parameter choice $(\bm\theta^\star,p^\star)$ of the channel model with
$p^\star\in(0,1)$ such that the two branches, $\rho_{(0)}$ and $\rho_{(1)}$ (Eq.~\eqref{eq:min_channel_def_split}), induce different local
quantum responses at both output sites:
\begin{equation}\label{eq:main_assumption_theo_1}
    \Delta_a
    :=
    m_a(1)-m_a(0)
    \neq 0,
    \qquad
    \Delta_b
    :=
    m_b(1)-m_b(0)
    \neq 0, ~~~~\mathrm{where}~~~m_a(s)
    :=
    \Tr[\hat Z_a\rho_{(s)}],
    ~~
    m_b(s)
    :=
    \Tr[\hat Z_b\rho_{(s)}].
\end{equation}

Let $\mathcal{E}^\star
    :=
    \mathcal{E}_{(\bm\theta^\star,p^\star)}$, and $Q^\star
    :=
    P_{\mathcal E_{(\bm\theta^\star,p^\star)}}$
be the Born distribution obtained by measuring $\mathcal{E}^\star$ in the
computational basis. Then for any depth $1<D< \frac{\operatorname{dist}(a,b)}{4}$
\begin{equation}\label{eq:equation_b2}
    % Q^\star
    % \in
    % \mathcal P_{\rm ch}(D),\qquad
% \mathrm{but}\qquad
    Q^\star
    \notin
    \mathcal P_{\rm unit}(D)
\end{equation}
% where $v$ is the maximum light-cone growth per circuit layer. For brickwall circuits with nearest-neighbor two-qubit gates $v=2$.

Moreover, there is a constant $\delta^\star=
    \frac{1}{6}
    p^\star(1-p^\star)
    |\Delta_a\Delta_b|
    >0$, such that the infimum of the total variation distance between $Q^*$ and the unitary distribution $P_{U(\bm \theta)}$ is lower bounded by
\begin{equation}
    \inf_{\bm\theta}
    \TV\!\left(
        Q^\star,
        P_{U(\bm\theta)}
    \right)
    \geq
    \delta^\star
\end{equation}
for every such depth $D$. In particular, if $\operatorname{dist}(a,b)=\Theta(N)$, then any purely unitary
nearest-neighbor model that represents $Q^\star$ requires linear depth, $D=\Omega(N)$.

\end{theorem}
%%%%%%%%%%
In Corollary~\ref{cor:cluster_angle_flip_endpoint_separation}, we will give explicit shallow-depth MBQC constructions of $Q^*$ satisfying Eq.~\eqref{eq:equation_b2} on a cluster state. There, we will show that already depth-$2$ in the channel model is sufficient to obtain such distributions $Q^*$.

%%%%%%%%%%
\begin{proof}
We prove the statement by exhibiting a two-qubit marginal that is generated by the correlated channel at some fixed shallow depth but cannot be generated by any unitary circuit whose backward light cones around the two relevant sites are disjoint.
%%%%%%%%%%%%

Write the branches of the channel model $\mathcal E^*$ as $U^{(s)}(\theta)=U_{D}(\theta_D)\cdots U_{l+1}(\theta_{l+1})P^s_{\mathcal{M}}\cdots U_1(\theta_1)$. Here, $P^s_{\mathcal{M}}$ is the stochastic Pauli string, acting on qubits $a$ and $b$, inserted with probability $p^*$ at an intermediate layer (compare Def.~\ref{def:min_channel}), and $D$ denotes the depth of $\mathcal E^*$. Let $s\in\{0,1\}$, with
\begin{equation}
    \Pr(s=1)=p^
    \star,
    \qquad
    \Pr(s=0)=1-p^
    \star.
\end{equation}
such that the channel model can be written as 
\begin{align}
    \mathcal{E}_{\theta^*,p^*}(\rho_{\rm in})&=\mathcal{E}^*=\Pr(s=0)U^{(0)}(\bm \theta)\rho_{\rm in}U^{(0)}(\bm \theta)^{\dagger}+\Pr(s=1)U^{(1)}(\bm \theta)\rho_{\rm in}U^{(1)}(\bm \theta)^{\dagger} \\
    &= (1-p^*)\rho_{(0)}+p^*\rho_{(1)}
\end{align}

Because the branch variable $s$ is classical, expectation values in the averaged channel are obtained by averaging the branch-conditioned expectation values. For the one-site observable $\hat{Z}_a$

\begin{equation}
    \Tr[\hat{Z}_a\mathcal{E}^*]=\langle\hat{Z}_a\rangle_{\mathcal{E}^*}=(1-p^*)\Tr[\hat{Z}_a\rho_{(0)}]+p^*\Tr[\hat{Z}_a\rho_{(1)}]
\end{equation}
and similarly
\begin{equation}
    \Tr[\hat{Z}_b\mathcal{E}^*]=\langle\hat{Z}_b\rangle_{\mathcal{E}^*}=(1-p^*)\Tr[\hat{Z}_b\rho_{(0)}]+p^*\Tr[\hat{Z}_b\rho_{(1)}]
\end{equation}
Define the branch-dependent local biases for both sites as
\begin{equation}
    m_a(s)
    :=
    \Tr[\hat{Z}_a\rho_{(s)}]=\langle\hat{Z}_a\rangle_s,
    \qquad
    m_b(s)
    :=
    \Tr[\hat{Z}_b\rho_{(s)}]=\langle\hat{Z}_b\rangle_s .
\end{equation}
Then
\begin{align}
    \langle\hat{Z}_a\rangle_{\mathcal{E}^*}=&(1-p^*)m_a(0)+p^*m_a(1),
    \qquad\langle\hat{Z}_b\rangle_{\mathcal{E}^*}=(1-p^*)m_b(0)+p^*m_b(1)
\end{align}

Regarding $\langle\hat{Z}_a\hat{Z}_b\rangle_{\mathcal{E}^*}$, within each fixed branch, the two relevant light-cone regions are disjoint at the insertion layer. 
Since the state on these regions factorizes, the conditional two-point expectation factorizes:
\begin{equation}
    \langle \hat{Z}_a\hat{Z}_b\rangle_s
    =
    \langle \hat{Z}_a\rangle_s
    \langle \hat{Z}_b\rangle_s
    =
    m_a(s)m_b(s).
\end{equation}
Therefore
\begin{equation}
\langle\hat{Z}_a\hat{Z}_b\rangle_{\mathcal{E}^*}=(1-p^*)m_a(0)m_b(0)+p^*m_a(1)m_b(1)
\end{equation}

The covariance of the two-site marginal of $Q^\star$ is hence
\begin{align}
    \Cov_{\mathcal{E}^*}(\hat Z_a,\hat Z_b)
    &=
    \langle \hat{Z}_a\hat{Z}_b\rangle_{\mathcal{E}^*}
    -
    \langle \hat{Z}_a\rangle_{\mathcal{E}^*}
    \langle \hat{Z}_b\rangle_{\mathcal{E}^*}
    \\
    &=
    (1-p^\star)m_a(0)m_b(0)
    +
    p^\star m_a(1)m_b(1)
    \\
    &\quad
    -
    \Bigl((1-p^\star)m_a(0)+p^\star m_a(1)\Bigr)
    \Bigl((1-p^\star)m_b(0)+p^\star m_b(1)\Bigr)
    \\
    &=
    p^\star(1-p^\star)
    \bigl[m_a(1)-m_a(0)\bigr]
    \bigl[m_b(1)-m_b(0)\bigr].
\end{align}
Thus the averaged channel has nonzero two-site covariance whenever
\begin{equation}
    \bigl[m_a(1)-m_a(0)\bigr]
    \bigl[m_b(1)-m_b(0)\bigr]
    \neq 0
\end{equation}
and $p^\star\in(0,1)$. We assume that the correlated branch changes both local responses, namely
\begin{equation}\label{eq:main_assumption}
    \Delta_a:=m_a(1)-m_a(0)\neq 0,
    \qquad
    \Delta_b:=m_b(1)-m_b(0)\neq 0 .
\end{equation}
% This assumption is satisfied whenever the inserted $Z^s$ byproduct is propagated through a later parametrized gate that anticommutes with $Z$ and the corresponding effective state has a nonzero response in the conjugate direction. The cluster-state construction in Corollary~\ref{cor:cluster_angle_flip_endpoint_separation} gives an explicit realization.

Now consider any purely unitary depth-$D$ circuit from the corresponding local nearest neighbor architecture with
\begin{equation}
    D<\frac{\operatorname{dist}(a,b)}{4}.
\end{equation}

By the light-cone bound, the backward light cones of the two measured sites $a$ and $b$ are disjoint (Lemma~\ref{lem:light-cone}).  Since the input state is a product state, according to Lemma~\ref{lem:disjoint_lightcone_factorized_marginal}
\begin{equation}
    \Cov_{U(\bm{\theta})}(\hat Z_a,\hat Z_b)=0
\end{equation}
for every parameter choice $\theta$ and every such depth $D$. Thus the joint marginal distribution on $(a,b)$ for the unitary model is a product distribution (see Lemma~\ref{lem:disjoint_lightcone_factorized_marginal})
\begin{equation}
    P_{U(\bm \theta)}(x_a,x_b)=P_{ab}=P_{U(\bm \theta)}(x_a)P_{U(\bm \theta)}(x_b),\qquad x_i\in\{0,1\}
\end{equation}
whereas under the condition $\Delta_a,\Delta_b\neq 0$, the marginal over $(a,b)$ for the channel distribution $Q^*$ can not be written as a product distribution for any $q(x_a)$ and $q(x_b)$ as below
\begin{equation}
    Q^*(x_a,x_b)=Q^*_{ab}\neq q(x_a)q(x_b),\qquad x_i\in\{0,1\}
\end{equation}
Therefore no such unitary distribution can have the same two-site marginal as $Q^\star$.  By the marginal-obstruction Lemma~\ref{lem:marginal_obstruction_classical}, no unitary distribution can equal the full distribution $Q^\star$.  Consequently,
\begin{equation}
    Q^\star\notin \mathcal P_{\rm unit}(D)
\end{equation}
for every
\begin{equation}
    D<\frac{\operatorname{dist}(a,b)}{4}.
\end{equation}

It remains to justify the quantitative total-variation gap.  Let $Q^\star_{ab}$ be the two-site marginal of $Q^\star$, and let $P_{ab}$ be any factorized two-site marginal generated by a unitary circuit below light-cone overlap, i.e., $D<\frac{\dist(a,b)}{4}$.  From Lemma~\ref{lem:covariance_tv_gap} one can write

\begin{equation}
    \left|
        \Cov_{\mathcal{E}^*}(\hat Z_a,\hat Z_b)
        -
        \Cov_{U(\theta
        )}(\hat Z_a,\hat Z_b)
    \right|
    \leq
    6\TV(Q^\star_{ab},P_{ab}).
\end{equation}
Since $\Cov_{U(\bm \theta)}(\hat Z_a,\hat Z_b)=0$, we obtain
\begin{equation}
    \TV(Q^\star_{ab},P_{ab})
    \geq
        \frac{1}{6}
    p^\star(1-p^\star)
    |\Delta_a\Delta_b|
    .
\end{equation}
Total variation is contractive under marginalization, so for the full distributions (see Lemma~\ref{lemma:marginal_obstruction})
\begin{equation}
    \TV(Q^\star,P_{U(\bm \theta)})
    \geq
    \TV(Q^\star_{ab},P_{U(\bm \theta),ab}).
\end{equation}
Thus
\begin{equation}
    \inf_\theta
    \TV(Q^\star,P_{U(\bm \theta)})
    \geq
    \delta^\star,
\end{equation}
where one may take
\begin{equation}
    \delta^\star
    :=
    \frac{1}{6}
    p^\star(1-p^\star)
    |\Delta_a\Delta_b|
    >0.
\end{equation}
The positivity follows from $p^*\in (0,1)$ and $\Delta_a\Delta_b\neq 0$. The gap is maximized at $p^\star=1/2$, giving
\begin{equation}
    \delta^\star_{\max}
    =
    \frac{1}{24}|\Delta_a\Delta_b|.
\end{equation}
Since $|\Delta_a|,|\Delta_b|\leq 2$ for Pauli observables, the universal upper bound is
\begin{equation}
    \delta^\star_{\max}
    \leq
    \frac{1}{6}.
\end{equation}

Finally, if $\operatorname{dist}(a,b)=\Theta(N)$, then the condition for light-cone overlap implies that any corresponding purely unitary nearest-neighbor realization of $Q^*$ requires depth $D=\Omega(N)$.
\end{proof}

\begin{corollary}[Separation for finite-range circuits]
\label{cor:finite_range_architecture_separation}
The separation in Theorem~\ref{thm:constant_depth_channel_separation} is not specific to the particular nearest-neighbor brickwall ansatz considered above. 
Consider any local circuit architecture in which the past light cone of any measured qubit expands by at most $vD$ sites after depth $D$, for some constant velocity $v=\mathcal O(1)$ independent of $N$. Then there exist output distributions generated by the correlated channel model at a fixed shallow depth $D$, independent of $N$, that cannot be generated by the corresponding purely unitary model at any depth
\begin{equation}
    D < \frac{\operatorname{dist}(a,b)}{2v},
\end{equation}
where $a$ and $b$ are the two sites linked by the stochastic Pauli string, and $\operatorname{dist}(a,b)$ denotes their distance in the circuit layout (see Def.~\ref{def:graph_dist}).

\end{corollary}
% Consequently, whenever $a$ and $b$ are separated by a distance growing with system size, the channel model, at shallow depth, realizes distributions that require growing depth in the corresponding purely unitary local circuit architecture. 
%In particular, if the distance between $a$ and $b$ 

In particular, the growth of $\dist(a,b)$ with system size induces lower bounds on the depths of the unitary competitors that can reproduce the shallow channel model distributions. Corollary~\ref{cor:finite_range_architecture_separation} also recovers the linear lower bound for one-dimensional nearest-neighbor brickwall circuits where $v=2$.
\begin{corollary}[No covariance separation with independent random Paulis]
\label{cor:no_covariance_separation_independent_paulis}
Consider the setting of Theorem~\ref{thm:constant_depth_channel_separation}.
Instead of the stochastic Pauli string $P_a^{s}\otimes P_b^{s}$ controlled by a shared variable $s\in\{0,1\}$, consider independent stochastic Pauli operations $P_a^{s_a}$ and $P_b^{s_b}$ on two distant sites $a$ and $b$,
where $s_a,s_b\in\{0,1\}$ are independent binary variables with $\Pr(s_a=1)=p_a$,
    and 
    $\Pr(s_b=1)=p_b$. Assume that, within each fixed branch $(s_a,s_b)\in\{0,1\}^2$, the two relevant light-cone
regions factorize, as in Theorem~\ref{thm:constant_depth_channel_separation}.
Then the averaged output state of the channel model satisfies 
\begin{equation}
        \Cov_{\mathcal E^\star}(\hat Z_a,\hat Z_b)=0 .
\end{equation}
Consequently, the two-site computational-basis marginal factorizes as $Q^\star_{ab}(x_a,x_b)
    =
    Q^\star_a(x_a)Q^\star_b(x_b),
    ~~
    x_a,x_b\in\{0,1\}.$
\end{corollary}

\begin{proof}
For fixed branch values $(s_a,s_b)$, define
\begin{equation}
    U^{s_a,s_b}(\bm\theta)
    =
    U_D(\bm \theta_D)\cdots U_k(\bm \theta_k)
    \left(P_a^{s_a}\otimes P_b^{s_b}\right)
    \cdots U_1(\bm \theta_1),
\end{equation}
and for some input state $\rho_{\rm in}$
\begin{equation}
    \rho_{(s_a,s_b)}
    =
    U^{s_a,s_b}(\bm\theta)
    \rho_{\rm in}
    U^{s_a,s_b}(\bm\theta)^\dagger .
\end{equation}
The averaged output state is then

\begin{align}
    \mathcal {E^\star}
    &=
    \sum_{s_a,s_b\in\{0,1\}}
    \Pr(s_a)\Pr(s_b)\rho_{(s_a,s_b)} 
\end{align}

Define the branch-dependent local responses
\begin{equation}
    m_a(s_a)
    :=
    \Tr[\hat Z_a\rho_{(s_a,s_b)}],
    \qquad
    m_b(s_b)
    :=
    \Tr[\hat Z_b\rho_{(s_a,s_b)}].
\end{equation}
Here $m_a$ depends only on $s_a$ and $m_b$ depends only on $s_b$, because the
two branches are independent and act on disjoint relevant regions. By the branch-dependent factorization assumption,
\begin{equation}
    \Tr[\hat Z_a\hat Z_b\rho_{(s_a,s_b)}]
    =
    \Tr[\hat Z_a\rho_{(s_a,s_b)}]
    \Tr[\hat Z_b\rho_{(s_a,s_b)}]
    =
    m_a(s_a)m_b(s_b).
\end{equation}
Therefore,

\begin{align}
    \langle \hat Z_a\rangle_{\mathcal E^\star}
    &=
    \sum_{s_a,s_b}
    \Pr(s_a)\Pr(s_b)m_a(s_a)=\sum_{s_a}
    \Pr(s_a)m_a(s_a)\sum_{s_b}
    \Pr(s_b) \\
    &=
    \sum_{s_a}
    \Pr(s_a)m_a(s_a),
\end{align}
since $m_a$ has no $s_b$ dependence and similarly $\langle \hat Z_b\rangle_{\mathcal E^\star}
    =
    \sum_{s_b}
    \Pr(s_b)m_b(s_b).$ For the two-site expectation,

\begin{align}
    \langle \hat Z_a\hat Z_b\rangle_{\mathcal E^\star}
    &=
    \sum_{s_a,s_b}
    \Pr(s_a)\Pr(s_b)m_a(s_a)m_b(s_b) \\
    &=
    \left(
        \sum_{s_a}
        \Pr(s_a)m_a(s_a)
    \right)
    \left(
        \sum_{s_b}
        \Pr(s_b)m_b(s_b)
    \right) \\
    &=
    \langle \hat Z_a\rangle_{\mathcal E^\star}
    \langle \hat Z_b\rangle_{\mathcal E^\star}.
\end{align}

Hence
\begin{equation}
    \Cov_{\mathcal E^\star}(\hat Z_a,\hat Z_b)
    =
    \langle \hat Z_a\hat Z_b\rangle_{\mathcal E^\star}
    -
    \langle \hat Z_a\rangle_{\mathcal E^\star}
    \langle \hat Z_b\rangle_{\mathcal E^\star}
    =
    0.
\end{equation}
By Lemma~\ref{lem:disjoint_lightcone_factorized_marginal}, the two-site
computational-basis marginal therefore factorizes:
\begin{equation}
    Q^\star_{ab}(x_a,x_b)
    =
    Q^\star_a(x_a)Q^\star_b(x_b).
\end{equation}
\end{proof}
Thus independent random Pauli byproducts do not generate the non-zero two-site
covariance used in Theorem~\ref{thm:constant_depth_channel_separation}. Hence
the covariance-based marginal obstruction does not yield a strict separation
for this independent-byproduct model. Nevertheless, the weak inclusion
\begin{equation}
    \mathcal P_{\rm unit}(D)
    \subseteq
    \mathcal P_{\rm ch}(D)
\end{equation}
still holds whenever the channel model contains the unitary model as the
special case with trivial byproduct probabilities.

\subsection{MBQC on a cluster state}\label{app:mbqc_cluster_proof}

\begin{figure}[t]
\centering
\includegraphics[width=0.4\textwidth]{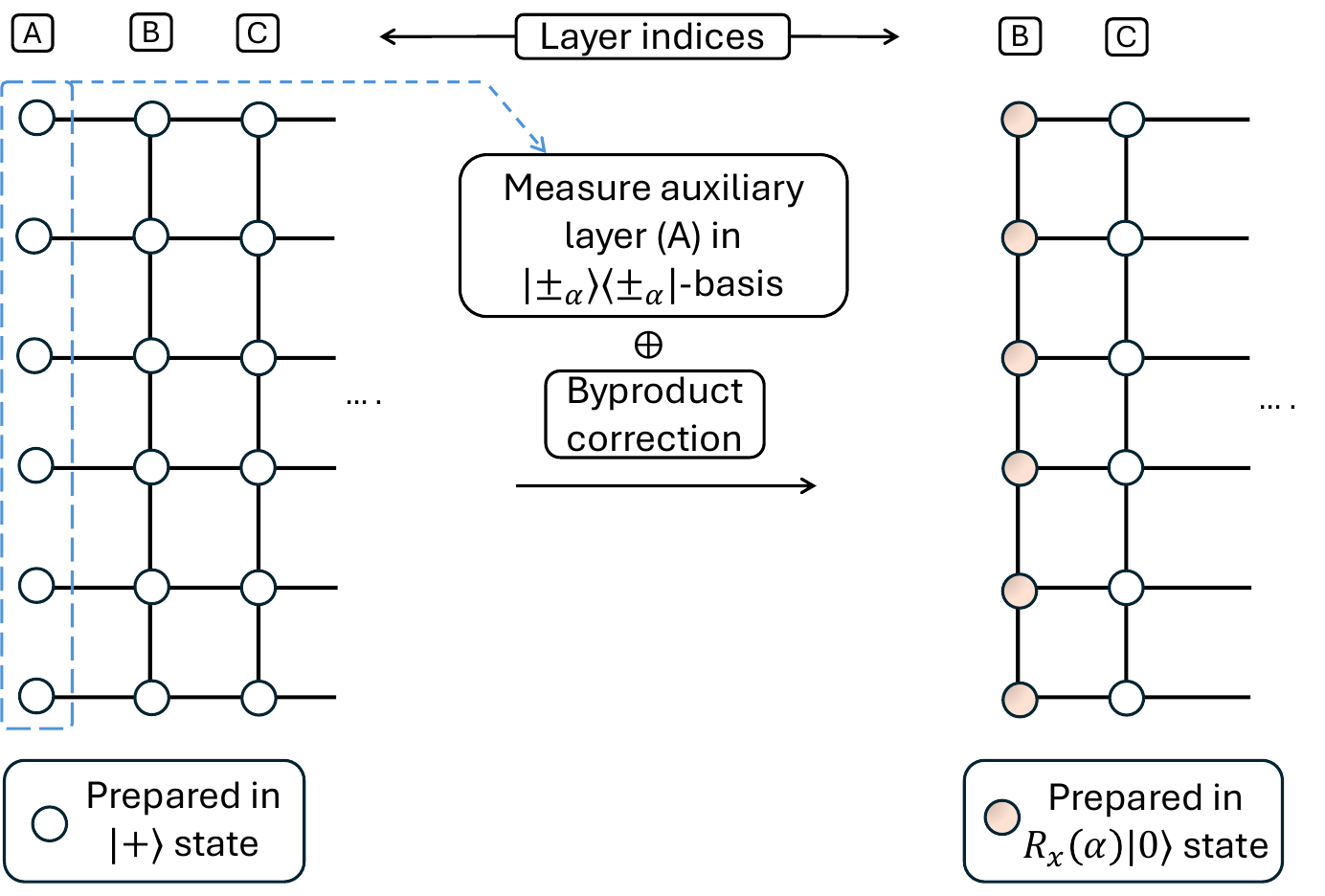}
\caption{\textbf{MBQC preparation of product input state.}
An auxiliary column $A$, initialized in $\ket{+}^{\otimes N}$, is attached to the first computational column. Measuring each auxiliary qubit in the $\{\ket{\pm_\alpha}\}$ basis and applying the corresponding byproduct correction prepares $R_x(\alpha)\ket{0}$ on every qubit of the first computational column, yielding $\ket{\psi_{\rm in}}=\left(R_x(\alpha)\ket{0}\right)^{\otimes N}$.} 
\label{fig:input_st_prep}
\end{figure}

As an illustrative realization of a one-dimensional nearest-neighbor architecture used in Theorem~\ref{thm:constant_depth_channel_separation}, we consider measurement-based quantum computation on an $N\times(D+1)$ cluster state, where the final column is the readout layer. We focus on an open chain of $N\geq6$ qubits with two computational layers, corresponding to depth $D=2$ (see Fig.~\ref{fig:cluster_ex}~(a)). The input is the product state of the form
\begin{align}
\rho_{\rm in}
=
\left(
\Rx(\alpha)\ket{0}\bra{0}\Rx(\alpha)^\dagger
\right)^{\otimes N},
\qquad
\sin\alpha\cos\alpha\neq0, \quad\alpha\notin\frac{\pi}{2}\mathbb{Z}
\end{align}
where $R_x(\alpha)=\exp{(-i\alpha X/2)}$. Because this state is separable, in the equivalent circuit picture, it can be prepared using local single-qubit rotations before the depth-$2$ entangling circuit. Its preparation therefore does not contribute to the circuit depth considered in the separation proof.

The same input state can be prepared directly within the cluster-state construction. Consider an $N\times(D+1)$ cluster state augmented by an auxiliary input column, denoted by column $A$ in Fig.~\ref{fig:input_st_prep} (left panel). All qubits are initialized in $\ket{+}$, and each auxiliary qubit is coupled by a $CZ$ gate to the corresponding qubit in the first computational column $B$ in Fig.~\ref{fig:input_st_prep}. The auxiliary qubits are then measured in the basis $\ket{\pm_{\alpha}}=\left(\ket{0}\pm e^{-i\alpha}\ket{1}\right)/\sqrt{2}$.

For the measurement outcome $s_j\in\{0,1\}$ on the auxiliary qubit associated with site $j$, the state teleported to the site $j$ in the first computational column $B$ is, up to an irrelevant global phase,
\begin{equation}
H Z^{s_j} R_z(\alpha)\ket{+}
=
X^{s_j} H R_z(\alpha)\ket{+}
=
X^{s_j} R_x(\alpha)\ket{0}.
\end{equation}
Applying the standard byproduct correction therefore prepares $R_x(\alpha)\ket{0}$ at site $j$ in the computational layer $B$. Repeating this procedure independently for all $j=1,\ldots,N$ prepares
\begin{equation}
\ket{\psi_{\rm in}}
=
\left(
R_x(\alpha)\ket{0}
\right)^{\otimes N}
\end{equation}
on the first computational column.

This construction is included solely to establish compatibility with the MBQC architecture and is not part of the depth-separation argument.

% \end{remark}

Performing $XY$-plane measurements on the qubits in the computational layers implements, as shown in Sec.~\ref{sec:mbqc_realizing_sh_ran}, the depth-$2$ unitary
\begin{align}\label{eq:app_cluster_unit}
&\qquad ~~~~~~~~~~~~~~~~~~~~~~~~~~~~~~~~~U(\bm{\theta})=\prod_{l=1}^2
    U_l(\bm\theta_l)
    \qquad \\
    \mathrm{where}\qquad &U_l(\bm\theta_l)=
    H^{\otimes N}\CZ_{\open}\Rz(\bm\theta_l),
    \qquad
    \CZ_{\open}:=\prod_{j=1}^{N-1}\CZ_{j,j+1},\qquad \Rz(\bm\theta_l)
    :=
    \prod_{j=1}^{N}\Rz(\theta_{j}^l)=\exp\!\left(-i\theta_j^l Z_j/2.\right)\nonumber
\end{align}

As discussed in Sec.~\ref{sec:mbqc_realization_of_co_ran}, shared classical randomness can be generated in MBQC by choosing the feedforward rules such that the same effective measurement-outcome bit controls byproducts on spatially separated qubits. For instance, the same stochastic bit acts on the boundary qubits $(1, 6)$ for a total $N=6$ qubits in Fig.~\ref{fig:cluster_ex}.

We next show how these correlated byproducts produce an output distribution that belongs to the correlated-channel family but cannot be generated by the corresponding shallow unitary model.

\begin{figure}[t]
\centering
\includegraphics[width=0.8\textwidth]{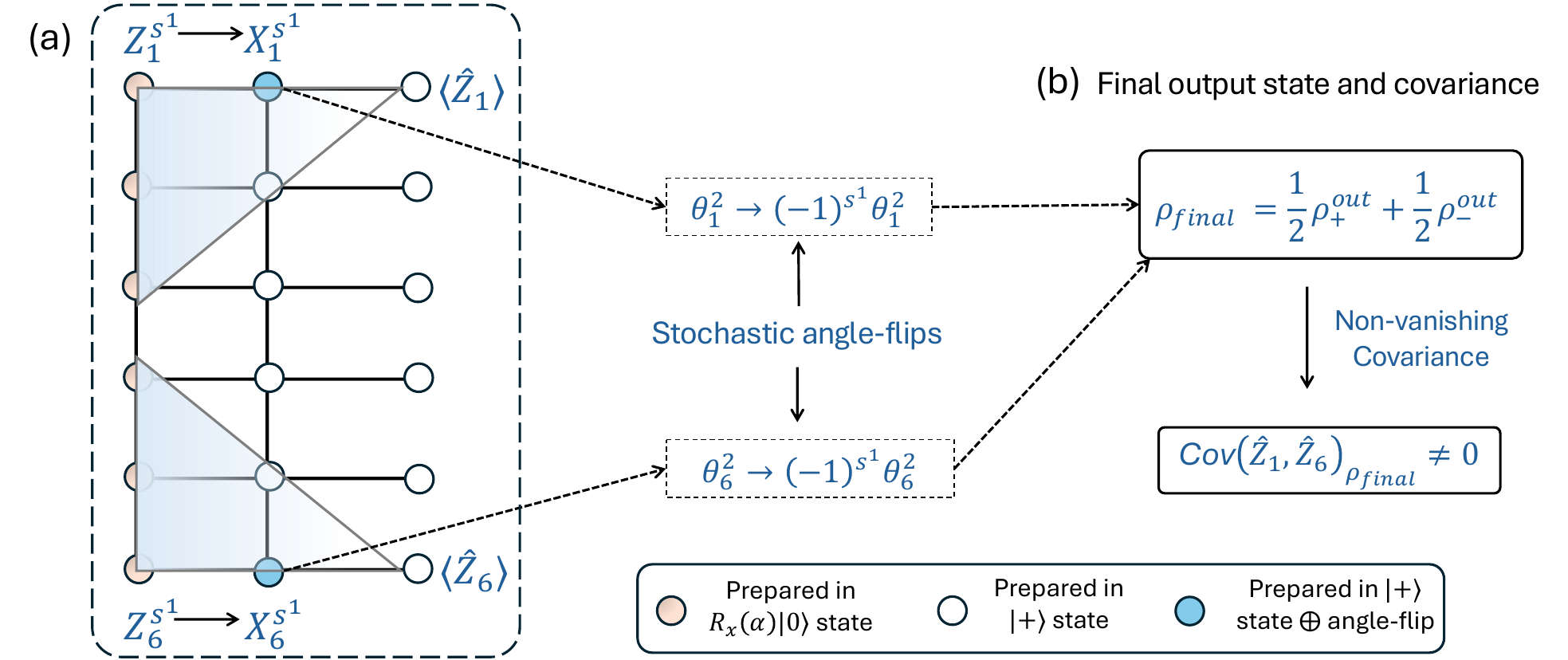}
\caption{\textbf{MBQC realization of the shallow-depth endpoint separation.}
(a) For illustration, we consider a $6\times3$ cluster state whose final column forms the readout layer. Classical feedforward is chosen so that the measurement-induced byproducts on the two boundary qubits of the first computational layer share the same binary variable $s^1\in\{0,1\}$, yielding correlated $Z_1^{s^1}$ and $Z_6^{s^1}$. Propagating these byproducts through the subsequent layer induces correlated sign flips of the corresponding rotation angles, $\theta_1^2\mapsto(-1)^{s^1}\theta_1^2$ and $\theta_6^2\mapsto(-1)^{s^1}\theta_6^2$.
(b) Averaging over the two stochastic branches $\rho_{+}^{\mathrm{out}}$ and $\rho_{-}^{\mathrm{out}}$ corresponding to $s^1=0$ and $s^1=1$, respectively, gives the final output state $\rho_{\mathrm{final}}=\frac{1}{2}\rho_{+}^{\mathrm{out}}+\frac{1}{2}\rho_{-}^{\mathrm{out}}$, whose endpoint observables exhibit a nonvanishing covariance, $\operatorname{Cov}(\hat Z_1,\hat Z_6)_{\rho_{\mathrm{final}}}\neq0$. The white qubits are initialized in $\ket{+}$, while the orange qubits are prepared in $R_x(\alpha)\ket{0}$. The blue qubits indicate the locations at which the shared byproduct induces the stochastic angle flips. This construction realizes the separation established in Corollary~\ref{cor:cluster_angle_flip_endpoint_separation}.} 
\label{fig:cluster_ex}
\end{figure}

\begin{corollary}
[Shallow-depth separation using MBQC on a cluster state]
\label{cor:cluster_angle_flip_endpoint_separation}

Consider the VMBQC channel model $\tilde{\mathcal E}_c$ introduced in Sec.~\ref{sec:end_pauli_corr} in which the propagated Pauli byproducts are corrected before the final measurements. We implement this model on the $N\times3$ cluster state shown in Fig.~\ref{fig:cluster_ex}~(a), with $N\geq6$, computational depth $D=2$, and open boundary conditions. Suppose that the effective measurement-induced byproducts associated with the two boundary qubits, $Z_1^{s^1}$ and $Z_N^{s^1}$, in the first computational layer are controlled by the same binary variable $s^1\in\{0,1\}$, sampled according to $\Pr(s^1=1)=p$, with $p\in(0,1)$ (see Sec.~\ref{sec:mbqc_realization_of_co_ran}).

Let $\rho_{(s^1)}$ denote the branch-conditioned output state corresponding to $s^1\in\{0,1\}$. Then the averaged output state is
\begin{equation}
    \mathcal{E}^\star
    :=
    (1-p)\rho_{(0)}+p\rho_{(1)}
\end{equation}
Let $Q^\star$ be the computational-basis Born distribution of
$ \mathcal{E}^\star$. For compactness, define
$c_{l,j}:=\cos(\theta^l_{j})$ and $s_{l,j}:=\sin(\theta^l_{j})$.
If the parameters of the model are chosen such that
\begin{equation}  
    \beta
    :=
    4p(1-p)
    s_{2,1}s_{2,N}s_{1,1}c_{1,2}c_{1,N-1}s_{1,N}
    \sin^4(\alpha)\cos^2(\alpha)
    \neq 0,
\end{equation}
then the two boundary sites have nonzero covariance when measured in $\hat Z$ basis,
\begin{equation}  
    \Cov_{\mathcal{E}^*}(\hat Z_1,\hat Z_{N})
    =
    \beta .
\end{equation}
Consequently, the infimum of the total variation distance between $Q^\star$ and the output distribution generated by the corresponding shallow unitary model satisfies
\begin{equation}  \inf_\theta\TV\!\left(Q^\star,P_{U(\bm{\theta})}\right)
    \geq
    \frac{|\beta|}{6}.
\end{equation}
In particular, the total variation bound above holds against every purely unitary model of depth $1<D<\frac{N-1}{2}$ under the cluster-state layer convention used above.
Hence, when the boundary separation scales as $\Theta(N)$, the required unitary depth scales as $D=\Omega(N)$.
\end{corollary}

\begin{proof}
Here, we consider the VMBQC $\tilde{\mathcal{E}}_c(\bm \theta,\bm p)$ model as described in Sec.~\ref{sec:end_pauli_corr}, in which all Pauli byproducts propagated to the end of the computation are corrected before the final measurement. Now consider the cluster state shown in Fig.~\ref{fig:cluster_ex} with $N\geq 6$, and suppose that the boundary qubits in the first computational layer acquire the correlated measurement-induced byproducts $Z_1^{s^1}$ and $Z_N^{s^1}$, as described in Sec.~\ref{sec:mbqc_realization_of_co_ran}. Using the commutation relations given in Eq.~(1) of Ref.~\cite{poulsen2024measurement}, one can show that the $Z^{s^1}_1$ byproduct becomes $X^{s^{1}}_1$ in the second layer, and therefore changes the corresponding measurement angle of the $R_z(\theta^2_1)$ gate as $X^{s^1}_1R_z(\theta^{2}_1)\mapsto R_z((-1)^{s^1}\theta^2_1)X^{s^1}_1$. In the next propagation step, $X^{s^{1}}_1$ becomes $Z^{s^{1}}_1X^{s^{1}}_2$ at the readout layer. Since these final Pauli byproducts are corrected as shown in Sec.~\ref{sec:end_pauli_corr}, the only remaining effect of $Z_1^{s^1}$ is the intermediate sign change of $\theta^{2}_1\mapsto(-1)^{s^1}\theta^2_1$. The same argument applies to the other boundary byproduct $Z_N^{s^1}$, whose only remaining effect is $\theta_N^2\mapsto(-1)^{s^1}\theta_N^2$.

Thus, for any $N\geq6$, the shared variable $s^1$ changes only the two boundary angles $\theta_1^2$ and $\theta_N^2$ in the second computational layer, and the rest of the $\{\theta^2_{j}\}_{j\in [2,N-1]}$ remain unaffected as shown in Fig.~\ref{fig:cluster_ex} (with cluster state width $N=6$).

For the input state $\Rx(\alpha)\ket{0}$, with $\Rx(\alpha)=\exp\left(-i\alpha X/2\right)$, on each qubit, the relevant Bloch
components are
\begin{equation}  
    \langle X\rangle=0,
    \qquad
    \langle Y\rangle=-\sin\alpha,
    \qquad
    \langle Z\rangle=\cos\alpha .
\end{equation}

Let's consider the output state after the first layer in the unitary model
\begin{equation}    \rho_1:=U_1(\bm\theta_1)\rho_{\rm in}U_1(\bm\theta_1)^\dagger .
\end{equation}

Now, we compute the branch-dependent endpoint responses. The
$\hat{Z}_1$ endpoint observable propagated backward through the second layer $U_2^{(s^1)}(\bm{\theta}_2)$ that depends on the branch $s^1\in\{0,1\}$ as

\begin{align}
    U_2^{(s^1)}(\bm{\theta}_2)^\dagger \hat{Z}_1 U_2^{(s^1)}(\bm{\theta}_2)
    =
    \Rz(\bm\theta_2^{(s^1)})^\dagger
    \CZ_{\open}^\dagger \hat{X}_1\CZ_{\open}
    \Rz(\bm\theta_2^{(s^1)}) 
    =
    \left(c_{2,1}\hat{X}_1-(-1)^{s^1} s_{2,1}\hat{Y}_1\right)\hat{Z}_2 .
\end{align}
where $\bm\theta_2^{(s^1)}$ denotes the set of measurement angles in the second layer with $(\theta_1^2,\theta_N^2)\mapsto((-1)^{s^1}\theta_1^2,(-1)^{s^1}\theta_N^2)$, while the rest of the angles remains unaffected by $s^1$. 
Hence, the branch-dependent final expectation value for $\hat{Z}_1$ is
\begin{align}
    m_1(s^1)
    :=
    \Tr[\hat{Z}_1\rho_{(s^1)}]
    =
    c_{2,1}A_1-(-1)^{s^1} s_{2,1}B_1,
\end{align}
where
\begin{align}
    A_1:=\Tr[(\hat{X}_1\hat{Z}_2)\rho_1],
    \qquad
    B_1:=\Tr[(\hat{Y}_1\hat{Z}_2)\rho_1].
\end{align}
Similarly,
\begin{align}
    m_{N}(s^1)
    :=
    \Tr[\hat{Z}_{N}\rho_{(s^1)}]
    =
    c_{2,N}A_{N}-(-1)^{s^1} s_{2,N}B_{N},
\end{align}
where
\begin{align}
    A_{N}:=\Tr[(\hat{Z}_{N-1}\hat{X}_{N})\rho_1],
    \qquad
    B_{N}:=\Tr[(\hat{Z}_{N-1}\hat{Y}_{N})\rho_1].
\end{align}

Direct Pauli propagation through the first cluster layer gives
\begin{align}
    A_1=s_{1,2}\sin\alpha\cos\alpha,
    \qquad
    B_1=-s_{1,1}c_{1,2}\sin^2\alpha\cos\alpha,
\end{align}
and
\begin{align}
    A_{N}=s_{1,N-1}\sin\alpha\cos\alpha,
    \qquad
    B_{N}=-c_{1,N-1}s_{1,N}\sin^2 \alpha\cos\alpha.
\end{align}

For each fixed branch $\rho_{(s^1)}$ with fixed $s^1$, the endpoint backward light cones are disjoint for
$N\geq 6$, at depth $D=2$. Since the input state is a product state, the branch-conditioned
endpoint moment factorizes:
\begin{align}
    \Tr[\hat{Z}_1\hat{Z}_{N}\rho_{(s^1)}]
    =
    m_1(s^1)m_{N}(s^1).
\end{align}
Therefore the averaged one-site responses are
\begin{align}
    \Tr[\hat{Z}_1\mathcal{E}^*]
    =
    \sum_{s^1\in\{0,1\}}\Pr(s^1)m_1(s^1),
    \qquad
    \Tr[\hat{Z}_{N}\mathcal{E}^*]
    =
    \sum_{s^1\in\{0,1\}}\Pr(s^1)m_{N}(s^1),
\end{align}
and the averaged two-site response is
\begin{align}
    \Tr[\hat{Z}_1\hat{Z}_{N}\mathcal{E}^*]
    =
    \sum_{s^1\in\{0,1\}}\Pr(s^1)m_1(s^1)m_{N}(s^1).
\end{align}
with this, we obtain

\begin{align}
    \Cov_{\mathcal{E}^*}(\hat{Z}_1,\hat{Z}_{N})
    &=
    \Tr[\hat{Z}_1\hat{Z}_{N}\mathcal{E}^*]
    -
    \Tr[\hat{Z}_1\mathcal{E}^*]\Tr[\hat{Z}_{N}\mathcal{E}^*] \\
    &=
    4p(1-p)s_{2,1}s_{2,N}B_1B_{N} \\
    &=
    4p(1-p)
    s_{2,1}s_{2,N}s_{1,1}c_{1,2}c_{1,N-1}s_{1,N}
    \sin^4 \alpha\cos^2 \alpha \\
    &=
    \beta .
\end{align}

Thus, $\beta\neq0$ implies that the two boundary outcomes have nonzero covariance. By Lemma~\ref{lem:disjoint_lightcone_factorized_marginal}, their computational-basis marginal therefore cannot factorize into the product of its single-qubit marginals.

By contrast, consider the corresponding purely unitary model obtained from MBQC on the same cluster state as in Fig.~\ref{fig:cluster_ex} by fully adapting all the measurement outcomes. This removes the stochastic correlated Pauli byproducts and yields the unitary circuit $U(\bm{\theta})$ in Eq.~\eqref{eq:app_cluster_unit} associated with the cluster-state computation. For the depths under consideration, the backward light cones of the two boundary measurements are disjoint. Thus, according to Lemma~\ref{lem:disjoint_lightcone_factorized_marginal} one can write
\begin{equation}
    \Cov_{\rho_U}(\hat{Z}_1,\hat{Z}_{N})=0.
\end{equation}
Consequently, the computational-basis marginal on the two boundary qubits factorizes into the product of its single-qubit marginals.

Applying the covariance--total-variation bound (see Lemma~\ref{lem:covariance_tv_gap}) to the two endpoint marginals $Q^\star_{1,N}$ and $P_{1,N}$ for the channel and unitary model respectively,
\begin{equation}
    \TV\!\left(Q^\star_{1,N},P_{1,N}\right)
    \geq
    \frac16
    \left|
        \Cov_{\mathcal{E}^*}(\hat{Z}_1,\hat{Z}_{N})
        -
        \Cov_{\rho_U}(\hat{Z}_1,\hat{Z}_{N})
    \right|
    =
    \frac{|\beta|}{6}.
\end{equation}
Since total variation cannot increase under marginalization,
\begin{equation}
    \inf_{\theta}\TV\!\left(Q^\star,P_{U(\bm{\theta})}\right)
    \geq
    \inf_{\theta}\TV\!\left(Q^\star_{1,N},P_{1,N}\right)
    \geq
    \frac{|\beta|}{6}.
\end{equation}
$\beta$ is maximized in absolute value by choosing the angle factors to have unit
magnitude, namely
\begin{equation}
    |s_{2,1}|=|s_{2,N}|=|s_{1,1}|=|s_{1,N}|=1,
    \qquad
    |c_{1,2}|=|c_{1,N-1}|=1 .
\end{equation}
Thus the remaining optimization is over
$   f(\alpha):=\sin^4 \alpha\cos^2 \alpha .
$
Writing $x=\sin^2 \alpha$, we have
\begin{equation}
    f(\alpha)=x^2(1-x),
    \qquad
    x\in[0,1].
\end{equation}
The maximum occurs at
$
    x=\frac23,~
    \text{i.e.},~
    \sin^2 \alpha=\frac23,~
    \cos^2 \alpha=\frac13,
$
and gives $\max_\alpha \sin^4 \alpha\cos^2 \alpha
    =
    \frac{4}{27}$.
Therefore, for fixed $p$, $|\beta|_{\max}
    =
    \frac{16}{27}p(1-p)$. If one also optimizes over the branch probability, the maximum is attained at
$p=1/2$, giving
\begin{equation}
    |\beta|_{\max}
    =
    \frac{4}{27}.
\end{equation}
Therefore no such purely unitary local circuit can reproduce $Q^\star$.
For the cluster-state layer convention used here, each layer expands an endpoint
backward light cone by at most one site, and thus the endpoint backward light cones
are disjoint whenever $D<(N-1)/2$. Therefore, if the endpoint distance is $\Theta(N)$, any
purely unitary nearest-neighbor realization requires depth $D=\Omega(N)$.
\end{proof}

The above proof applies to a cluster state with open boundary conditions, but it can be extended to periodic boundary conditions as well. 

\section{Distribution classes}\label{app:dist_classes}

This appendix clarifies the structured target family of distributions used in the main
text. They are not intended to characterize the full expressive power of the
channel model. Rather, they isolate simple distributional structure that
are naturally produced by shared branch randomness.

\subsection{Branch-peaked mixture distributions}
\label{app:branch_peaked_mixture}

A branch-peaked mixture distribution has the form
\begin{equation}
    P(x)
    =
    \sum_{\bm s\in\mathcal S}
    \pi(\bm s)
    \left[
        \delta_{\bm s}\,\mathbf 1[x=x_{\bm s}^\star]
        +
        (1-\delta_{\bm s})R_{\bm s}(x)
    \right],
    \qquad x\in\{0,1\}^N ,
    \label{eq:app_branch_peaked}
\end{equation}
where $\mathcal S\subseteq\{0,1\}^L$ is a finite set satisfying
$|\mathcal S|\leq S_{\max}$, with $S_{\max}=O(1)$ is independent of $N$. Here,
$\pi(\bm s)\ge 0$ is the probability of the corresponding branch $\bm s$ with
\begin{equation}
    \sum_{\bm s\in\mathcal S}\pi(\bm s)=1,
\end{equation}
$\delta_{\bm s}\in(0,1]$, and
$x_{\bm s}^\star\in\{0,1\}^N$ is the dominant bit string associated with
branch $\bm s$. Each $R_{\bm s}$ is a normalized residual distribution on
$\{0,1\}^N$ and is allowed to depend on the branch.

The residual $R_{\bm s}$ describes the probability mass of branch $\bm s$
away from its dominant bit string. More explicitly, suppose the output
distribution of branch $\bm s$ is $P_{\bm s}$. If
$x_{\bm s}^\star$ is chosen as a dominant bit string of $P_{\bm s}$ and $\delta_{\bm s}:=P_{\bm s}(x_{\bm s}^\star),$ then, for $\delta_{\bm s}<1$, one can write
\begin{equation}
    P_{\bm s}(x)
    =
    \delta_{\bm s}\,\mathbf 1[x=x_{\bm s}^\star]
    +
    (1-\delta_{\bm s})R_{\bm s}(x),
    \label{eq:branch_residual_decomposition}
\end{equation}
with
\begin{equation}
    R_{\bm s}(x)
    =
    \begin{cases}
        0, & x=x_{\bm s}^\star, \\[1ex]
        \dfrac{P_{\bm s}(x)}{1-\delta_{\bm s}}, & x\neq x_{\bm s}^\star .
    \end{cases}
    \label{eq:residual_distribution}
\end{equation}
This makes $R_{\bm s}$ a normalized distribution:
\begin{equation}
    \sum_{x\in\{0,1\}^N} R_{\bm s}(x)
    =
    \frac{1}{1-\delta_{\bm s}}
    \sum_{x\neq x_{\bm s}^\star}P_{\bm s}(x)
    =
    1 .
\end{equation}
Thus Eq.~\eqref{eq:app_branch_peaked} is simply the branch average $P=\sum_{\bm s\in\mathcal S}\pi(\bm s)P_{\bm s}$ written in a way such that the dominant contribution of each branch is explicitly
separated from the remaining branch-dependent fluctuations, which is denoted as the residual distribution $R_{\bm s}$ associated with branch $\bm s$. When
$\delta_{\bm s}=1$, the branch is a point mass at $x_{\bm s}^\star$, and the
residual term does not contribute. To use this as a nontrivial peaked class, one should assume a peaked regime,
for example
\begin{equation}
    \delta_{\bm s}\ge \delta_{\min}>0,
    \qquad
    R_{\bm s}(x_{\bm s}^\star)=0,
    \qquad
    \text{for all }\bm s\in\mathcal S .
\end{equation}

An explicit example of such a distribution is provided in Eq.~\eqref{eq:num_example_explicit}. This class is naturally suited to the correlated channel learning model, in which shared randomness, in the form of stochastic Pauli strings, appears in the middle of the circuit as shown in Eq.~\eqref{eq:app_channel_model} and in Appendix~\ref{def:input_byp_model}. In
those cases, the sampled branch can change the effective circuit evolution, so different branches may have
different dominant bit strings and different residual shapes.

\section{Correlated channel model with Pauli-correction}
\label{def:input_byp_model}

\begin{figure*}[t]
\centering
\includegraphics[width=1\textwidth]{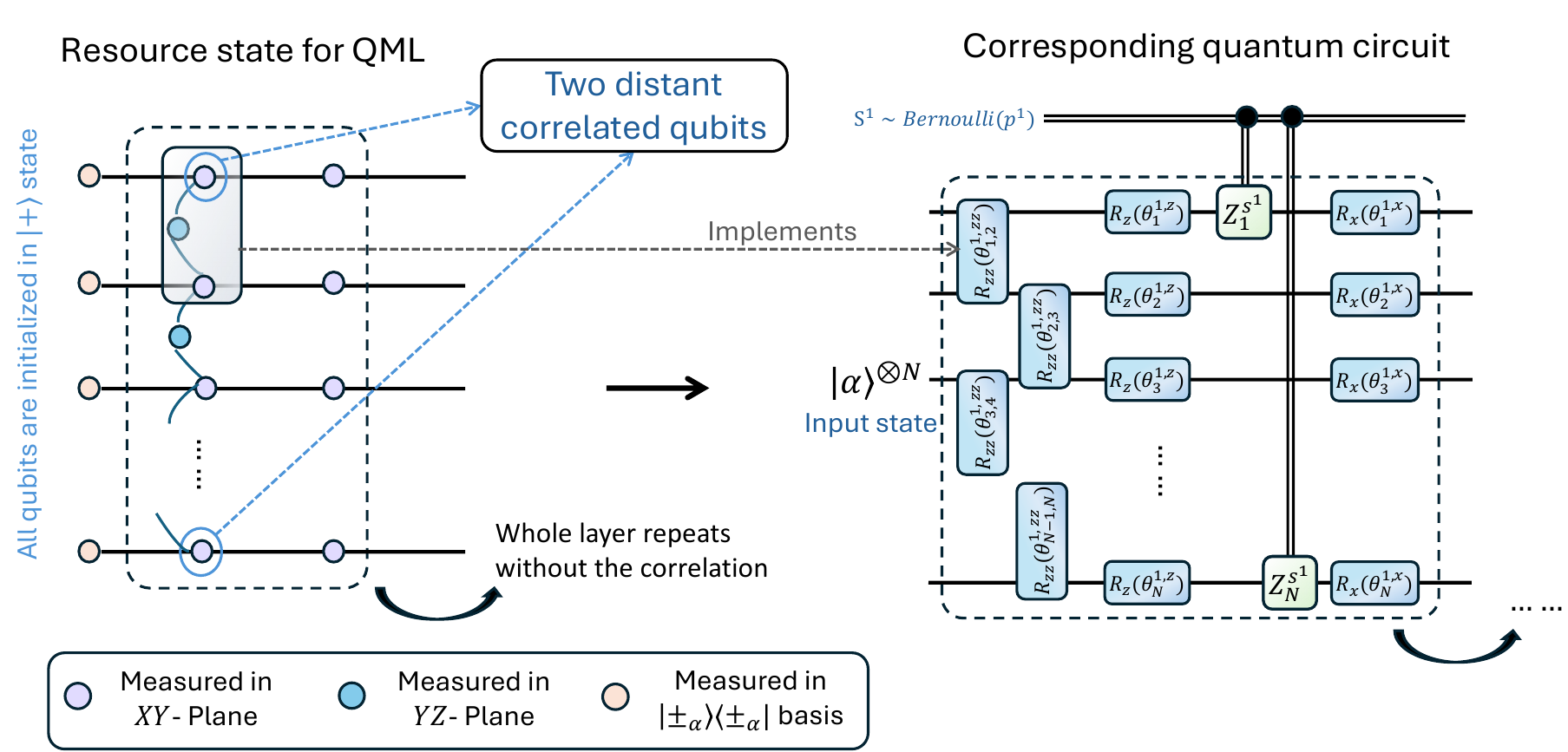}
\caption{\textbf{Resource state and learning model:} 
The left panel shows an open-boundary graph-state resource, where vertices represent qubits and edges denote $CZ$ entangling gates. All resource qubits are initialized in $\ket{+}$ and subsequently measured in the color-coded basis indicated in the legend. Measuring the qubits within the dashed region implements one layer of the circuit shown on the right, consisting of nearest-neighbor $R_{zz}$ gates, local $R_z$ rotations, a correlated stochastic $Z$-byproduct layer, and local $R_x$ rotations. Repeating the dashed graph pattern generates additional circuit layers. Classical feedforward is chosen so that the two distant boundary qubits, in the dashed box, share the same effective byproduct bit $s^1\in \{0,1\}$, with $\Pr(s^1=1)=p^1\in[0,1]$ as described in Sec.~\ref{sec:mbqc_realization_of_co_ran}. Consequently, the measurement-induced byproducts either occur jointly or are both absent, corresponding to the stochastic Pauli string $Z_1^{s^1}\otimes Z_N^{s^1}$, at the designated location in the circuit picture. The input-measurement pattern prepares the effective circuit input state $\ket{\alpha}^{\otimes N}$, where $\ket{\alpha}=R_x(\alpha)\ket{0}$ and $\alpha\notin\frac{\pi}{2}\mathbb{Z}$.} 
\label{fig:model_num}
\end{figure*}

As a concrete realization of the learning model with shared randomness in Sec.~\ref{sec:mbqc_based_model}, we consider the MBQC resource state shown in the left panel of Fig.~\ref{fig:model_num}. With an appropriate choice of adaptive single-qubit projective measurements, this resource state implements the one-dimensional nearest-neighbor circuit with open boundary conditions shown in the right panel. The single-qubit rotations in each circuit layer are realized using standard $XY$-plane MBQC measurement patterns~\cite{raussendorf2001one}, while the two-qubit $R_{zz}$ rotations are implemented using appropriate graph-state gadgets involving $YZ$-plane measurements~\cite{qin2024applicability,PhysRevLett.132.220602}. Moreover, as described in Sec.~\ref{sec:mbqc_realization_of_co_ran}, after collecting the measurement outcomes from layer $l$ of the resource state, classical feedforward can be used to implement correlated stochastic Pauli operations on a subset of qubits in that layer, controlled by a single shared random variable.

 Each layer in the corresponding induced circuit can be written as
\begin{equation}\label{eq:model_num}
U_k(\bm{\theta}_k)=\bigotimes_{i=1}^N R_x(\theta^{k,x}_{i}) \bigotimes_{i=1}^N R_z(\theta^{k,z}_{i})\prod_{i=1}^{N-1}R_{zz}(\theta^{k,zz}_{i,i+1}) 
\end{equation}
with layer parameter vector
\begin{equation} \label{eq:layer_params} \boldsymbol{\theta}_{k} = \left( \{\theta^{k,x}_{i}\}_{i=1}^N, \{\theta^{k,z}_{i}\}_{i=1}^N, \{\theta^{k,zz}_{i,i+1}\}_{i=1}^{N-1} \right). 
\end{equation}

For a chosen insertion point associated with layer $l$, we decompose the ordered circuit as
\begin{equation}
U(\bm{\theta})
=
U_{>l}(\bm{\theta})U_{\leq l}(\bm{\theta}),
\end{equation}
where $U_{\leq l}$ and $U_{>l}$ denote the circuit portions before and after the insertion slot, respectively. We sample a shared bit $s^l\in\{0,1\}$ with $\Pr(s^l=1)=p^l$ and insert the stochastic Pauli string
\begin{equation}
P_{\mathcal{M}}^{s^l}=Z_{\mathcal{M}}^{s^l}
:=
\bigotimes_{j\in M}Z_j^{s^l},
\end{equation}
at this slot, where $M\subseteq[1,N]$. Propagating $Z_{\mathcal{M}}^{s^l}$ through the remaining Pauli rotations induces branch-dependent sign flips of the angles whose generators anticommute with $Z_{\mathcal{M}}$, while leaving the Pauli string itself unchanged. After applying the final Pauli correction $\hat P(\bm s)^{\dagger}=(Z^{s^l}_{\mathcal{M}})^{\dagger}$ as described in Sec.~\ref{sec:end_pauli_corr} and shown in Fig.~\ref{fig:two_learning_models}, the corrected branch unitary becomes
\begin{equation}
U^{(s^l)}(\bm{\theta})
=
Z_{\mathcal{M}}^{s^l}
U_{>l}(\bm{\theta})
Z_{\mathcal{M}}^{s^l}
U_{\leq l}(\bm{\theta})
\end{equation}
Now, the final $Z_{\mathcal{M}}^{s^l}$ before measurements does not alter the computational-basis measurement outcomes. Thus, it can be omitted at the level of the computational-basis output probabilities. The resulting channel, for input state $\rho_{\rm in}=(\ket{\alpha}\bra{\alpha})^{\otimes N}$ with $\ket{\alpha}=R_x(\alpha)\ket{0}$, and $\alpha\notin\frac{\pi}{2}\mathbb{Z}$, is
\begin{align}\label{eq:num_channel}
\mathcal E_{\bm{\theta},p^l}(\rho_{\rm in})
=&
(1-p^l)
U^{(0)}(\bm{\theta})\rho_{\rm in}U^{(0)}(\bm{\theta})^\dagger
\nonumber\
+\\
&p^l
U^{(1)}(\bm{\theta})
\rho_{\rm in}
U^{(1)}(\bm{\theta})^\dagger,
\end{align}
and measurement in the computational $Z$ basis produces output distribution 
\begin{equation}
P_{\mathcal E_{\bm{\theta},p^l}}(x)
=
\Tr\left[
\ket{x}\bra{x}
\mathcal E_{\bm{\theta},p^l}(\rho_{\rm in})
\right]
\end{equation}

In Fig.~\ref{fig:model_num} (right), we choose $l=1$, and $M=\{1,N\}$. We deliberately place the stochastic Pauli string $P^{s^1}_{\mathcal{M}}=Z^{s^1}_1\otimes Z^{s^1}_N$ in the first layer before the sub-layer $\bigotimes_{i=1}^N R_x(\theta_i^{k,x})$.
Propagating this string through the remaining circuit induces the correlated angle flips of the corresponding $R_x$ angles, $\theta\mapsto(-1)^{s^1}\theta$, while leaving the $R_z$ and $R_{zz}$ rotations unchanged. Since the final byproduct remains diagonal in the computational basis, it does not alter the sampled bit string and need not be physically corrected. By contrast, choosing $X_{\mathcal{M}}^{s^1}=X_1^{s^1}\otimes X_N^{s^1}$ flips the angles of subsequent $R_z$ and $R_{zz}$ rotations whose generators anticommute with $X_{\mathcal{M}}$. Since the final $X_{\mathcal{M}}^{s^1}$ byproduct flips the corresponding computational-basis outputs, it is removed using the final Pauli correction shown in Fig.~\ref{fig:two_learning_models} and described in Sec.~\ref{sec:end_pauli_corr}. In both cases, the retained stochastic effect is the correlated branch-dependent angle-flip pattern rather than an output bit-flip map or classical post-processing.

\begin{figure}[t]
\centering
\includegraphics[width=0.85\textwidth]{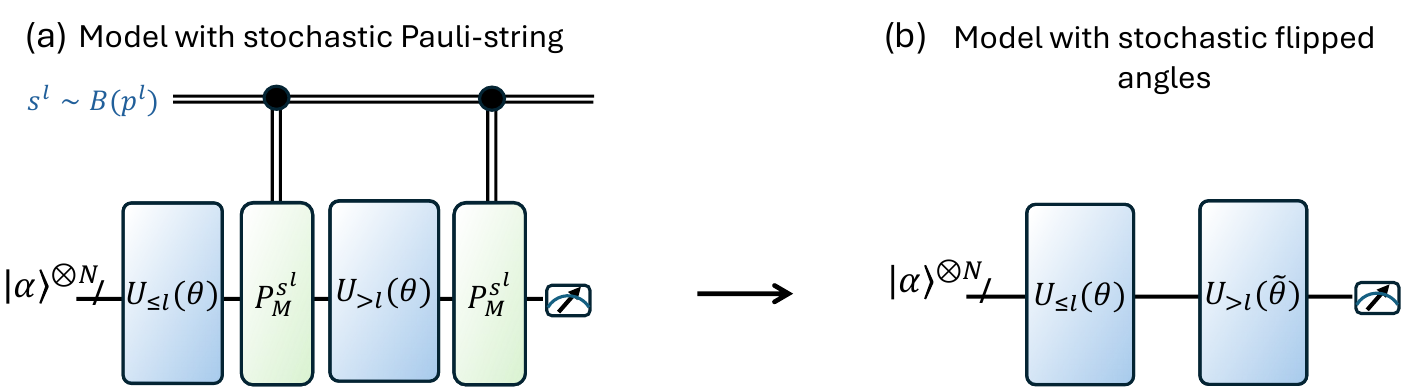}
\caption{\textbf{Learning model with Pauli correction:} Consider the unitary model with total $D$ layers where each layer consists of Pauli rotations as parameterized gates of the form $\exp(-i\theta Q_{i,i+1})$, acting on two qubits, where $Q\in \{I,X,Y,Z\}^{\otimes 2}$. (a) A stochastic Pauli string $P_{\mathcal{M}}^{s^l}$, where $P_{\mathcal{M}}=\bigotimes_{j\in M}P_j$, $P_j\in \{X,Z\}$, $M\subseteq[1,N]$, controlled by a shared variable $s^l\sim\operatorname{Bernoulli}(p^l)$, is inserted after layer $l$. The same $s^l$-dependent Pauli string is applied again at the output, with $s^l=0$ corresponding to the identity and $s^l=1$ to $P_{\mathcal{M}}$. (b) Propagating the inserted Pauli string through $U_{>l}(\bm\theta)$ leaves commuting gates unchanged and flips the signs of the angles of anticommuting Pauli rotations. Specifically, $P_{\mathcal{M}}^{s^l}\exp(-i\theta Q_{i,i+1})=\exp(-i(-1)^{s^l}\theta Q_{i,i+1})P_{\mathcal{M}}^{s^l}$ if $[P_{\mathcal{M}}^{s^l},Q_{i,i+1}]\neq 0$. The propagated Pauli string then cancels the final correction because $(P_{\mathcal{M}}^{s^l})^2=\id$. The resulting branch is therefore equivalent to the purely unitary circuit $U_{>l}(\widetilde{\bm\theta})U_{\leq l}(\bm\theta)$ with branch-dependent sign-flipped angles $\bm \theta\mapsto\widetilde{\bm\theta}$.} 
\label{fig:two_learning_models}
\end{figure}

\section{Numerical details}\label{app:num_details}
\subsection{Gradient of the loss function}\label{app:grad_details}

In this section, we derive the gradient of the squared MMD loss in Eq.~\eqref{eq:mmd loss} with respect to the application probabilities $\{p^i\}_{i=1}^L$. The output distribution of the channel model, in Eq.~\eqref{eq:channel_dist_all}, can be rewritten using Eq.~\eqref{eq:branch_prob_decomp} as 
\begin{align}\label{eq:loss_expanded}
    P_{\mathcal E_{\bm\theta,\bm p}}(x)
    =
    \sum_{\bm s\in\{0,1\}^L}
    \pi_{\bm p}(\bm s)\,
    \bra{x}\,
        U^{(\bm s)}(\bm\theta)
        \rho_{\rm in}
        U^{(\bm s)}(\bm\theta)^\dagger\ket{x}
        =\sum_{\substack{s^i\in\{0,1\}\\i\in\{1,\cdots,L\}}}
    \pi_{p^1}(s^1)\cdots\pi_{p^L}(s^L)
    P_{\bm \theta}(x|s^1,\cdots,s^L)
\end{align}

where $\bm s=(s^1,\ldots,s^L)\in\{0,1\}^L$ and $\pi_{p^i}(s^i)=(p^i)^{s^i}(1-p^i)^{1-s^i}$ and
\begin{equation}
P_{\bm\theta}(x\mid\bm s)
=
\bra{x}
U^{(\bm s)}(\bm\theta)
\rho_{\rm in}
U^{(\bm s)}(\bm\theta)^\dagger
\ket{x}
\end{equation}
is the output probability of bitstring $x$ conditioned on the branch $\bm s$. For each random variable $s^l\in \{s^i\}_{i=1}^L$
\begin{equation}
    s^l =
\begin{cases}
		0  & \mbox{with  probability } 1- p^l,\\
		1 & \mbox{with  probability } p^l.
\end{cases}
\end{equation}

Thus, for any variable $s^l$, Eq.~\eqref{eq:loss_expanded} can be written as

\begin{align}
    &P_{\mathcal E_{\bm\theta,\bm p}}(x)= \\ 
    &p^l\cdot \sum_{\substack{s^{i}, \, i\in \{1\dots L\}, \\ i\neq l}} \left[\prod_{i \ne l }\pi_{p^i}(s^i) \right]P_{\bm \theta}(x|s^1,\cdots,1,\cdots,s^L) + 
    \\
    & (1-p^l)\sum_{\substack{s^{i}, \, i\in \{1\dots L\}, \\ i\neq l}} \left[\prod_{i \ne l }\pi_{p^i}(s^i) \right]P_{\bm \theta}(x|s^1,\cdots,0,\cdots,s^L)
\end{align}

The rest of the calculation remains the same as in Appendix B of Ref.~\cite{majumder2024variational}. This gives the final gradient as

\begin{align}\label{eq:grad_p}
    \frac{\partial \mathcal{L}}{\partial p^l} &= 2 \mathop{\mathbb{E}}_{\substack{x \sim P_{\mathcal E_{\bm\theta,\bm p^{(l,1)}}}\\ y \sim P_{\mathcal E_{\bm\theta,\bm p}}}} [K(x,y)]
     - 2 \mathop{\mathbb{E}}_{\substack{x \sim P_{\mathcal E_{\bm\theta,\bm p^{(l,0)}}}\\ y \sim P_{\mathcal E_{\bm\theta,\bm p}}}} [K(x,y)]\nonumber
    \\
    &\ - 2  \mathop{\mathbb{E}}_{\substack{x \sim P_{\mathcal E_{\bm\theta,\bm p^{(l,1)}}}\\ y \sim Y}} [K(x,y)]
     + 2  \mathop{\mathbb{E}}_{\substack{x \sim P_{\mathcal E_{\bm\theta,\bm p^{(l,0)}}}\\ y \sim Y}} [K(x,y)],
\end{align}

where $\bm p^{(l,1)}$ and $\bm p^{(l,0)}$ denote the parameter vector $\bm p=(p^1,\ldots,p^L)$ with $p^l$ replaced by $1$ and $0$, respectively. Here, $Y$ is the target distribution.

Hence, the gradient with respect to $p^l$ can be estimated by evaluating the model twice: once with $p^l=1$ and once with $p^l=0$, while keeping all other parameters fixed.

The gradients with respect to the variational angles $\theta\in\bm\theta$ can be evaluated using standard parameter-shift rules for parametrized Pauli rotations~\cite{liu2018differentiable,majumder2024variational} as

\begin{align}\label{eq:grad_t}
    \frac{\partial \mathcal{L}}{\partial \theta^{l}_{j}} = {} & \mathop{\mathbb{E}}_{\substack{x \sim P_{\mathcal E_{\bm\theta^+,\bm p}}\\ y \sim P_{\mathcal E_{\bm\theta,\bm p}}}} [K(x,y)]-\mathop{\mathbb{E}}_{\substack{x \sim P_{\mathcal E_{\bm\theta^-,\bm p}}\\ y \sim P_{\mathcal E_{\bm\theta,\bm p}}}} [K(x,y)]\nonumber \\
      & -\mathop{\mathbb{E}}_{\substack{x \sim P_{\mathcal E_{\bm\theta^+,\bm p}}\\ y \sim Y}} [K(x,y)]+\mathop{\mathbb{E}}_{\substack{x \sim P_{\mathcal E_{\bm\theta^-,\bm p}}\\ y \sim Y}} [K(x,y)]
\end{align}

Here, $P_{\mathcal E_{\bm\theta^+,\bm p}}$ and $P_{\mathcal E_{\bm\theta^-,\bm p}}$ are the output distributions of the channel model with parameters $\bm{\theta^{\pm}} = \bm{\theta} \pm \frac{\pi}{2} \bm{e} ^{l}_{j}$ and $\bm{e}^{l}_{j}$ is a vector with a $1$ at position $(l,j)$ and $0$s elsewhere. The parameter shifting is done for each angle individually, while the other angles remain unchanged. For the unitary model, to estimate the gradient of the variational angles we set $\bm p=\bm 0$ as $P_{U(\boldsymbol{\theta})}
\equiv
P_{\mathcal{E}_{\bm\theta,\bm p=\bm 0}}$. Eq.~\eqref{eq:grad_p} and ~\eqref{eq:grad_t} can be used to compute the gradient of the learning model defined in Sec.~\ref{sec:mbqc_based_model}.

\subsection{Box plot details}\label{app:box_details}
Each box plot in Fig.~\ref{fig:results_1} summarizes the minimum MMD losses obtained from $20$
independent training runs. The losses are first ordered from smallest to
largest. The lower and upper boundaries of the box are the first and third
quartiles, denoted by $Q_1$ and $Q_3$. Thus, approximately $25\%$ of the
observations lie below $Q_1$, approximately $50\%$ lie between $Q_1$ and
$Q_3$, and approximately $25\%$ lie above $Q_3$. The box therefore
represents the central $50\%$ of the training outcomes, corresponding
approximately to $10$ of the $20$ runs. 

The interquartile range,
$\mathrm{IQR}=Q_3-Q_1$, measures the spread of the central half of the
data and is less sensitive to extreme values than the full range. The
lower whisker extends to the smallest observed loss that is not below
$Q_1-1.5\,\mathrm{IQR}$, while the upper whisker extends to the largest
observed loss that is not above $Q_3+1.5\,\mathrm{IQR}$. Observations
outside these limits are displayed individually as outliers (white circles). In a
standard box plot, the horizontal line inside the box denotes the median,
which separates the lower and upper halves of the data. In Fig.~\ref{fig:results_1}, the median is shown with the orange line.

\end{document}